\documentclass[journal,comsoc]{IEEEtran}
\usepackage{graphicx}
\usepackage{amsfonts}
\usepackage{amssymb}
\usepackage{amsmath}
\usepackage{amsthm}
\usepackage{cite}
\usepackage{subfigure}
\usepackage{mathrsfs}
\usepackage[displaymath,mathlines]{lineno}
\usepackage{color}
\usepackage{tabulary}
\usepackage{multirow}
\usepackage{algpseudocode}
\usepackage{algorithm,algpseudocode}
\usepackage{algorithmicx}
\usepackage{pbox}
\usepackage{multicol}
\usepackage{lipsum}
\usepackage{amsthm}
\usepackage{relsize}
\usepackage{lipsum}
\usepackage{epstopdf}
\usepackage{mathtools}
\usepackage{calc}
\usepackage{makecell}
\usepackage{cases}

\makeatletter
\newcommand{\doublewidetilde}[1]{{%
		\mathpalette\double@widetilde{#1}}}
\newcommand{\double@widetilde}[2]{%
	\sbox\z@{$\m@th#1\widetilde{#2}$}%
	\ht\z@=.5\ht\z@
	\widetilde{\box\z@}}
\makeatother

\usepackage{accents}
\newcommand*{\dt}[1]{%
	\accentset{\mbox{\large .}}{#1}}

\newtheorem{theorem}{Theorem}
\newtheorem{lemma}{Lemma}

\newtheorem{remark}{Remark}

\begin{document}
	\title{\huge RIS-Assisted Wireless Communications: Long-Term versus Short-Term Phase Shift Designs}
	\author{Trinh~Van~Chien,  \textit{Member}, \textit{IEEE}, Lam~Thanh~Tu,  Waqas~Khalid, Heejung~Yu, \textit{Senior Member}, \textit{IEEE},  Symeon~Chatzinotas, \textit{Fellow}, \textit{IEEE}, and Marco~Di~Renzo,~\textit{Fellow},~\textit{IEEE}
 \thanks{This research is funded by Hanoi University of Science and Technology (HUST) under project number T2022-TT-001. The work of M. Di Renzo was supported in part by the European Commission through the H2020 ARIADNE project under grant agreement number 871464 and through the H2020 RISE-6G project under grant agreement number 101017011, and by the Agence Nationale de la Recherche (ANR PEPR-5G and Future Networks, grant NF-
PERSEUS, 22-PEFT-004). Parts of this paper were presented at ICASSP~2022 \cite{van2022controlling}.}
 \thanks{T. V. Chien is with the School of Information and Communication Technology (SoICT), Hanoi University of Science and Technology, Vietnam (email: chientv@soict.hust.edu.vn).}
  \thanks{L. T. Tu is with Ton Duc Thang University, Vietnam (email: tulamthanh@tdtu.edu.vn).}
  \thanks{W. Khalid is with the Institute of Industrial Technology, Korea University, Sejong 30019, South Korea (email: waqas283@korea.ac.kr).}
  \thanks{H. Yu is with the Department of Electronics
and Information Engineering, Korea University, Sejong 30019, Korea
(email:heejungyu@korea.ac.kr).}
\thanks{S. Chatzinotas is with the Interdisciplinary Centre for Security, Reliability and Trust (SnT), University of Luxembourg, L-1855 Luxembourg, Luxembourg (email: symeon.chatzinotas@uni.lu).}
\thanks{M. Di Renzo is with Universit\'e Paris-Saclay, CNRS, CentraleSup\'elec, Laboratoire des Signaux et Syst\`emes, 3 Rue Joliot-Curie, 91192 Gif-sur-Yvette, France. (marco.di-renzo@universite-paris-saclay.fr)}
	}

	\maketitle

	\begin{abstract}
	Reconfigurable intelligent surface (RIS) has recently
	gained significant interest as an emerging technology for
	future wireless networks thanks to its potential for improving the  coverage probability in challenging propagation environments. This paper studies an RIS-assisted
	propagation environment, where a source transmits data to a destination in the presence of a weak direct link. We analyze and compare RIS designs based on long-term and short-term channel statistics in terms of coverage probability and ergodic rate. For the
	considered optimization designs, we derive closed-form expressions for the coverage probability and ergodic rate, which explicitly unveil the impact of both the propagation environment and the RIS on the system performance. Besides the optimization of the RIS phase profile, we formulate an RIS placement optimization problem with the aim of maximizing the coverage probability by relying only on partial channel state information. An efficient algorithm is  proposed based on the gradient ascent method. Simulation results are illustrated in order to corroborate the analytical framework and findings. The proposed RIS phase profile is shown to outperform several heuristic benchmarks in terms of outage probability and ergodic rate. In addition, the proposed RIS placement strategy provides an extra degree of freedom that remarkably improves  system performance.   
	\end{abstract}
	\IEEEpeerreviewmaketitle
	\section{Introduction}
	
Although fifth-generation (5G) networks are under deployment worldwide, the massive growth of the number of different users/devices and data traffic poses major challenges for beyond 5G or sixth-generation (6G) networks \cite{tataria20216g,tariq2020speculative}. Specifically, the number of mobile devices is estimated to reach 25 billions by 2025 \cite{jonsson2021ericsson} and the requested peak data rate per link may exceed 100 Gbps by 2030 \cite{tariq2020speculative}. Dealing with such challenging requirements in harsh propagation environments subject to the presence of several blocking objects (obstacles) and deep fading is still a challenging task to solve for wireless engineers and researchers. A major challenge in $5$G and beyond networks lies in improving the data rate and the coverage probability in a cost-effective, energy-sustainable, and economically viable manner. This is why mobile operators resort to different network nodes to offer blanket coverage in their deployments \cite{3gpp2020technical}. Within the third generation partnership program (3GPP), for example, different new network nodes are currently being discussed, including the integrated access and backhaul (IAB) node, the network-controlled repeater node, the reconfigurable intelligent surface (RIS) node, and the smart skin node \cite{flamini2022towards}. Recently, RISs have gained considerable attention from  academia and industry due to their ability of controlling the propagation characteristics of wireless environments via passive scattering elements integrated with low-cost and low-power electronics \cite{qian2020beamforming, wu2019intelligent, zappone2020overhead, kumar2022achievable,wei2021channel}. For example, the authors of \cite{qian2020beamforming}  studied the statistics of the signal-to-noise ratio (SNR) when the number of RIS elements grows large. The minimization of the total transmit power at a multiple antenna access point (AP) in multiple-input single-output (MISO) RIS-assisted wireless systems was investigated in \cite{wu2019intelligent}, by jointly optimizing the transmit beamforming at the AP and the phase profile at the RIS. The authors of \cite{zappone2020overhead} studied a single-user multiple-input multiple-output (MIMO) system with the help of a single RIS and optimized it by considering the overhead for channel estimation and feedback. The deployment of an RIS in cognitive radio systems was studied in \cite{kumar2022achievable}, where the spectrum is shared between the primary and secondary users. With the exception of \cite{qian2020beamforming}, which relies on large-scale analysis, these research works are focused on optimization algorithms and do not provide any closed-form expressions for the distribution of the signal-to-noise ratio (SNR) and ergodic rate of RIS-aided networks.

The complex nature of wireless environments results in propagation channels that are characterized by small-scale, i.e., short-term, and large-scale, i.e., long-term, fading. An RIS aims to shape the electromagnetic waves in complex wireless environments by appropriately optimizing the reflection coefficients of its constitutive elements (i.e., the unit cells) \cite{tang2020wireless}. Due to the small-scale and large-scale dynamics that characterize  complex wireless channels, the phase shifts of the RIS elements can be optimized based on different time scales \cite{zhi2021two,di2021communication}. The vast majority of research works in the literature have considered the optimization of the RIS phase profile based on  instantaneous channel state information (CSI) \cite{wu2019intelligent, zhang2020sum, zappone2020overhead,huang2019reconfigurable,wei2022joint,huang2021multi} and references therein for application to sub-6 gigahertz and terahertz frequency bands. Interested readers are referred to \cite[Sec. III-B]{pan2022overview} for an extensive overview of optimizing RISs based on different levels of CSI. The optimization criterion based on instantaneous CSI operates by adjusting the phase shifts of the RIS elements based on the channel's small-scale dynamics. Therefore, it results in the best achievable performance at the price of a significant channel estimation overhead \cite{van2022reconfigurable,zappone2020overhead, wei2021channel,pan2021overview}. For these reasons, this optimization criterion may not be applicable in scenarios that are characterized by a short coherent time since the optimal phase shifts of the RIS elements need to be updated frequently in order to adapt to the rapid changes that characterize the dynamics of small-scale fading \cite{jung2020performance, pan2022overview}. This, in fact, may not be possible if the number of RIS elements and/or the number of users is large since the associated channel estimation overhead may result in very few resources available for communication \cite{han2019large,zhao2021two,dai2021statistical,luo2021reconfigurable,you2021reconfigurable}.

In contrast to these aforementioned research works that assume the availability of instantaneous CSI, another option for optimizing the phase shifts of the RIS elements is based on leveraging only statistical CSI, i.e., the large-scale characteristics of the wireless channel \cite{han2019large, abrardo2020intelligent, 9140329,luo2021reconfigurable, zhao2021two, dai2021statistical,peng2021analysis,luo2021reconfigurable,you2021reconfigurable} that vary at a much longer time scale.  Optimization criteria based on long-term CSI need to be updated less frequently, and this reduces the channel estimation overhead \cite{hu2021two}. In most research works, however, the channel statistics are analyzed numerically \cite{gan2021ris}, which are difficult to use for system optimization \cite{trigui2022performance}. More recently, some research works, such as \cite{zhi2021two,dai2021statistical}, have  obtained closed-form expressions for the ergodic channel rates of RIS-aided Massive MIMO communications when the base station is equipped with many antennas. These research works are based on a lower bound for the rate that may  not be tight in  scenarios where channel hardening does not hold \cite{van2022reconfigurable}. Even though some research works have recently proposed and analyzed the design of RISs based on long-term or statistical CSI, to the best of our knowledge, no previous work has proposed an analytical framework to comprehensively analyze and compare the achievable performance of RIS-assisted wireless systems based on short-term and long-term CSI.

In addition to the impact of CSI on the optimal design of RISs, another major open research issue is the optimal deployment of RISs in wireless networks, so as to achieve the best performance. For example, the optimal deployment of RISs is extensively discussed in \cite{you2022deploy}. In the far field, RISs are usually deployed close to the transmitter or the receiver since this minimizes the path-loss function \cite{tang2020wireless}. If the receivers are not static, locating the RISs close to the transmitter is often convenient to ensure that the transmission distance of one of the two transmission links is minimized \cite{perovic2021maximum}. Recently, hybrid deployments with RISs located close to the transmitters and the receivers have been proposed by leveraging the secondary reflections from multiple deployed RISs \cite{zheng2021double,nguyen2022leveraging}. To the best of our knowledge, a comprehensive analysis of the optimal deployment of RISs under different levels of CSI has not been reported in the open technical literature. In this paper, we analyze this important research problem.  


Motivated by these considerations, we consider the deployment of an RIS in the coverage area of a  source and a destination, and compare the system performance under phase shift designs based on different time scales, i.e., levels of CSI, and study the optimal placement of the RIS for the different options.  Specifically, the main contributions made by this paper are summarized as follows:
   \begin{itemize}
   	\item We introduce a novel analytical framework for evaluating the performance of RIS-aided channels in which the phase profile and the location of the RIS are optimized. Two different phase shift design criteria are considered: The first, referred to as the short-term phase shift design, exploits CSI and offers the best performance at the cost of a higher computational complexity. The second, referred to as the long-term phase shift design, exploits only the channel statistics to optimize the phase shifts.
   	\item We derive closed-form expressions for the coverage probability and ergodic channel rate for the long-term or short-term phase shift design, by using the moment matching technique. Although the statistical parameters behave differently when the number of RIS elements is finite, we mathematically show that the coverage probability approaches unity in the limiting regime of an infinite number of RIS elements. 
   	\item We formulate and solve a coverage probability maximization problem as a function of the RIS location subject to  a given phase shift design. The gradient ascent method is utilized as a low-cost solution to solve the optimization problem and to  overcome its inherent non-convexity. 
   	\item Numerical results demonstrate the coverage improvement offered by deploying RISs in wireless networks. For a large number of phase shifts, the long-term phase shift design provides a coverage probability and an ergodic channel rate that are close to the optimal values obtained by relying on instantaneous CSI. Moreover, the optimization of the location of the RIS is shown to offer a significant performance improvement.
   \end{itemize}
	\textit{Notation}: Upper and lower bold symbols denote matrices and vectors, respectively. The superscripts $(\cdot)^H$ and $(\cdot)^T$ denote the Hermitian and transpose operations.  $\mathbb{E}\{ \cdot \}$ and $\mathsf{Var}\{ \cdot \}$ denote the expectation and the variance of a random variable.  $\mathcal{CN}(\cdot,\cdot)$ denotes a circularly symmetric complex Gaussian distribution. The upper incomplete Gamma function and the Gamma function are denoted by $\Gamma(m,n) = \int_n^\infty t^{m-1} \exp(-t) dt$ and $\Gamma(x) = \int_{0}^{\infty} t^{x-1} \exp(-t) dt$, respectively.  The generalized hypergeometric function is denoted by ${_p}F_q (\cdot, \cdot, \cdot)$ with ${_1}{F}_1 (\cdot, \cdot, \cdot)$, i.e., $p=q=1$ denoting the confluent hypergeometric function of the first kind, and the hypergeometric function ${_2}{F}_2 (\cdot, \cdot, \cdot)$, i.e.,  $p=q=2$. The notation $f(x) \asymp  g(x)$ means that two real functions $f(x)$ and $g(x)$ fulfill the condition  $\lim_{x \rightarrow \infty} f(x)/g(x) = 1$. The first-order  derivative of the function $f(x)$ is denoted as $\dt{f}(x)$.
	\section{System model}
	We consider an RIS-assisted communication system where a single-antenna source transmits data to a single-antenna destination. We assume a quasi-static fading model where the channel gains are static and frequency flat in each time slot. The communication link from the source to the destination is assumed to be weak, and therefore it is enhanced by the assistance of an RIS equipped with $M$ tunable reflecting elements. The phase shift matrix $\pmb{\Phi} \in \mathbb{C}^{M \times M}$ is denoted by $\pmb{\Phi} = \mathrm{diag}\big([e^{j\theta_{1}}, \ldots, e^{j\theta_{M}}]^T \big)$, where $\theta_{m} \in [-\pi, \pi]$ is the phase shift  of the $m$-th RIS element. 
	\subsection{Channel Model}
	The channel of the direct link between the source and the destination is denoted by $h_{\mathrm{sd}} \in \mathbb{C}$. The corresponding indirect link comprises the channel between the source and the RIS, denoted by $\mathbf{h}_{\mathrm{sr}} \in \mathbb{C}^M$, and the channel between the RIS and the destination, denoted by $\mathbf{h}_{\mathrm{rd}} \in \mathbb{C}^M$. The channel between the source and the destination is assumed not to have a dominant path (i.e., a line-of-sight component), while the source-RIS and RIS-destination channels are characterized by line-of-sight and non-line-of-sight components. Specifically, the channels are modeled as\footnote{The line-of-sight and non-line-of-sight components of the channels are determined by the transmission distance, the shadowing, and the multipath propagation in a general wireless environment \cite{rappaport2010wireless}.}
	\begin{align} \label{eq:ChannelMod}
		h_{\mathrm{sd}} = \sqrt{\beta_{\mathrm{sd}}} g_{\mathrm{sd}} \mbox{ and } \mathbf{h}_{\alpha} = \bar{\mathbf{h}}_{\alpha} + \mathbf{g}_{\alpha}, 
	\end{align}
	where $g_{\mathrm{sd}} \sim \mathcal{CN}(0,1)$ is the small-scale fading coefficient of the direct link between the source and the destination; and $\mathbf{g}_{\alpha} \sim \mathcal{CN}(\mathbf{0}, (\beta_{\alpha}/(\kappa_{\alpha}+1))  \mathbf{I}_M)$  with $\alpha \in \{\mathrm{sr}, \mathrm{rd} \}$  are the vectors of the small-scale fading coefficients of the RIS-aided indirect links; $\beta_{\mathrm{sd}}$ and $\beta_{\alpha}$ are the large-scale fading coefficients; $\kappa_{\alpha}$ are the Rician factors with $ \kappa_{\alpha} >0$. The deterministic LoS channel vectors $\bar{\mathbf{h}}_{\alpha} \in \mathbb{C}^M$ are modelled as
	\begin{equation} \label{eq:barhsr}
		\bar{\mathbf{h}}_{\alpha} = \sqrt{\frac{\kappa_{\alpha} \beta_{\alpha} }{\kappa_{\alpha}+1}} \left[e^{j \mathbf{k}(\psi_{\alpha}, \phi_{\alpha} )^T \mathbf{u}_1}, \ldots,  e^{j \mathbf{k}(\psi_{\alpha}, \phi_{\alpha} )^T \mathbf{u}_M} \right]^T, 
	\end{equation}
	where $\psi_{\alpha}$ and $\phi_{\alpha}$ are the azimuth and elevation angles of departure (AoD)  from  the source as seen  by  the RIS.  Since the RIS is a planar array (i.e., a surface), the wave vectors $\mathbf{k}(\psi_{\alpha}, \phi_\alpha)$ is defined as
	\begin{equation}
		\mathbf{k}(\psi_{\alpha}, \phi_\alpha) = \frac{2\pi}{\lambda} \left[ \cos(\psi_{\alpha})\cos(\phi_\alpha), \, \sin(\psi_{\alpha})\cos(\phi_\alpha), \,\sin(\psi_{\alpha}) \right]^T,
	\end{equation}
	where $\lambda$ is the signal wavelength. In \eqref{eq:barhsr}, furthermore, the indexing vector $\mathbf{u}_m$ is defined as $\mathbf{u}_m =[0, \, \mathrm{mod}(m-1,N_H), \, \lfloor (m-1)/N_H \rfloor]^T$. 
	Due to the presence of LoS and NLoS components, i.e., $\bar{\mathbf{h}}_{\alpha}$ and  $\mathbf{g}_{\alpha}$, respectively, in RIS-aided links, the analytical performance evaluation of the considered system model is not straightforward.
	\subsection{Data Transmission and Achievable Rate}
The source is assumed to transmit a data symbol $x$ subject to the unit-power constraint $\mathbb{E}\{ |x|^2 \} = 1$. Thus, the received signal $y \in \mathbb{C}$ at the destination is 
	\begin{equation} \label{eq:ReceivedSig}
		y = \sqrt{p} ( h_{\mathrm{sd}}  +  \mathbf{h}_{\mathrm{sr}}^H \pmb{\Phi} \mathbf{h}_{\mathrm{rd}} ) x  + w,
	\end{equation}
	where $p$ is the transmit power allocated to the data symbol $x$; and  $w$ is the additive noise, which is distributed as $w \sim \mathcal{CN}(0, \sigma^2)$ with $\sigma^2$ being the noise variance. 
 
  Conditioned on  the RIS phase shift matrix, we provide the second and fourth moments of the indirect link, i.e., the link from the source to the destination through the RIS, in Theorem~\ref{Theorem:RISChannel}. 
	\begin{theorem} \label{Theorem:RISChannel}
		For a given phase shift matrix $\pmb{\Phi}$, the indirect link from the source to the destination through the RIS has the following second and fourth moments: 
		\begin{align}
			\mathbb{E} \{ |\mathbf{h}_{\mathrm{sr}}^H \pmb{\Phi} \mathbf{h}_{\mathrm{rd}}|^2 \} & =   \delta , \label{eq:2MomentCascade}\\
			\mathbb{E} \{ |\mathbf{h}_{\mathrm{sr}}^H \pmb{\Phi} \mathbf{h}_{\mathrm{rd}}|^4 \} & =  \delta^2  + \tilde{a}, \label{eq:4MomentCascade}
		\end{align}
		where $\delta = |\bar{\alpha}|^2 + M\mu \tilde{\kappa}$, $\bar{\alpha} = \bar{\mathbf{h}}_{\mathrm{sr}}^H \pmb{\Phi} \bar{\mathbf{h}}_{\mathrm{rd}} $, $\mu = \beta_{\mathrm{sr}}\beta_{\mathrm{rd}}/ ((\kappa_{\mathrm{sr}}+1)(\kappa_{\mathrm{rd}}+1))$, $\widetilde{\kappa} = \kappa_{\mathrm{sr}}+ \kappa_{\mathrm{rd}}+1$, $\hat{\kappa} = 1+ 2\kappa_{\mathrm{sr}} + 2\kappa_{\mathrm{rd}}$, and $\tilde{a} =2M |\bar{\alpha}|^2 \mu  \tilde{\kappa} + M^2 \mu^2 \tilde{\kappa}^2 +  2M \mu^2 \hat{\kappa} +8 |\bar{\alpha}|^2 \mu$.
	\end{theorem}
	\begin{proof}
		See Appendix~\ref{Appendix:RISChannel}.
	\end{proof}
	The analytical expressions in \eqref{eq:2MomentCascade} and \eqref{eq:4MomentCascade} depend on the LoS and NLoS channels and unveil that the second and fourth moments scale with $M$ and $M^2$, respectively. Next, Theorem~\ref{Theorem:RISChannel} will be utilized to facilitate the performance evaluation of RIS-aided channels. 
	
	The phase-shift matrix $\pmb{\Phi}$ is usually optimized based on different levels of CSI \cite{pan2021overview}. In this paper, we focus our attention on two design criteria:
    \begin{itemize}
    	\item \textit{Long-term phase-shift design}: The phase shifts of the RIS
    	elements are optimized based on statistical CSI. Specifically, the optimal phase shift matrix is obtained by  maximizing the average SNR at the destination. For this design criterion, the phase shift matrix is denoted by  $\pmb{\Phi}^{\mathsf{lt}} =  \mathrm{diag}\big([e^{j\theta_{1}^{\mathsf{lt}}}, \ldots, e^{j\theta_{M}^{\mathsf{lt}}}]^T \big)$, where $-\pi \leq \theta_{m}^{\mathsf{lt}} \leq \pi$ is the phase shift of the $m$-th element of the RIS. In this case, in simple terms, the optimization of the RIS depends on the large-scale fading coefficients and the LoS components of the channel.
		\item \textit{Short-term phase-shift design}: The phase shifts of the RIS elements are optimized based on  instantaneous CSI, which encompasses large-scale and small-scale fading statistics. For this design criterion, the phase shift matrix is denoted by $\pmb{\Phi}^{\mathsf{st}} =  \mathrm{diag}\big([e^{j\theta_{1}^{\mathsf{st}}}, \ldots, e^{j\theta_{M}^{\mathsf{st}}}]^T \big)$, where $-\pi \leq \theta_{m}^{\mathsf{st}} \leq \pi$ is the phase shift  of the $m$-th element of the RIS.
	\end{itemize}
The short-term phase-shift design corresponds to the best achievable performance. However, it entails the largest channel estimation overhead. Therefore, it is usually not an effective approach if the number of RIS elements and/or the number of users is large. Conditioned on the phase shift matrix and the CSI, the rate is formulated as
\begin{equation} \label{eq:Rate}
R = \log_2 ( 1 + \gamma)  \mbox{ [b/s/Hz]}, 
\end{equation}
where the SNR, denoted by $\gamma$, is
\begin{subnumcases}
	{\gamma =}
	\nu |h_{\mathrm{sd}} + \mathbf{h}_{\mathrm{sr}}^H \pmb{\Phi}^{\mathrm{lt}} \mathbf{h}_{\mathrm{rd}}|^2, & \mbox{long-term}, \label{eq:SNRLTv}\\
	\nu |h_{\mathrm{sd}} + \mathbf{h}_{\mathrm{sr}}^H \pmb{\Phi}^{\mathrm{st}} \mathbf{h}_{\mathrm{rd}}|^2,&  \mbox{short-term}. \label{eq:SNRSTv}
\end{subnumcases}
with $\nu = p/\sigma^2$. In the following lemma, we provide closed-form expressions for the RIS phase-shift matrix that correspond to the long-term and short-term design criteria, and that maximize the SNR at the destination.
\begin{lemma} \label{Corollary:LTOpt}
	Consider the long-term phase shift design criterion. If the phase of the $m$-th RIS element is set as follows
	\begin{equation} \label{eq:Optphasse}
		\theta_{m}^{\mathsf{opt}, \mathsf{lt}} = -\arg([\bar{\mathbf{h}}_{\mathrm{sr}}^\ast]_m)  - \arg([\bar{\mathbf{h}}_{\mathrm{rd}}]_m) ,
	\end{equation}
	where $[\bar{\mathbf{h}}_{\mathrm{sr}}^\ast]_m$ and $[\bar{\mathbf{h}}_{\mathrm{rd}}]_m$ are the $m$-th elements of the LoS components $\bar{\mathbf{h}}_{\mathrm{sr}}^\ast$ and $\bar{\mathbf{h}}_{\mathrm{rd}}$, respectively, then the average received SNR is maximized.
	
	Consider the short-term phase shift design criterion. If the phase of the $m$-th RIS element is set as follows
	\begin{equation} \label{eq:OptSt}
		\theta_{m}^{\mathsf{opt}, \mathsf{st}} = \arg(h_{\mathrm{sd}}) - \arg([\mathbf{h}_{\mathrm{sr}}^\ast]_m)  - \arg([\mathbf{h}_{\mathrm{rd}}]_m), \forall m,
	\end{equation}
	where $[\mathbf{h}_{\mathrm{sr}}^\ast]_m$ and $[\mathbf{h}_{\mathrm{rd}}]_m$ are the $m$-th elements of the RIS-aided channels $\mathbf{h}_{\mathrm{sr}}^\ast$ and $\mathbf{h}_{\mathrm{rd}}$, respectively, then the instantaneous received SNR is maximized. 
	
\end{lemma}
\begin{proof}
	See Appendix~\ref{Appendix:LTOpt}. 
\end{proof}
The key contribution of Lemma~\ref{Corollary:LTOpt} lies in the closed-form expressions for the optimal phase shift designs based on different levels of CSI. Specifically, the short-term and long-term
phase shift designs in \eqref{eq:Optphasse} and \eqref{eq:OptSt} are computed based on different time scales, but they are both aimed at boosting the strength of
the received signal. Based on the optimal phase shift designs in Lemma~\ref{Corollary:LTOpt}, the corresponding optimal rates are given as
\begin{equation} \label{eq:OptRate}
	R^{\mathrm{opt}} = \log_2 ( 1 + \gamma^{\mathrm{opt}}) \mbox{ [b/s/Hz]}, 
\end{equation}
where the optimal SNR $\gamma^{\mathrm{opt}}$ is
\begin{subnumcases} 
	{\gamma^{\mathrm{opt}} =}
	\nu |h_{\mathrm{sd}} + \mathbf{h}_{\mathrm{sr}}^H \pmb{\Phi}^{\mathsf{opt}, \mathsf{lt}} \mathbf{h}_{\mathrm{rd}}|^2, & \mbox{long-term}, \label{eq:OptSNRLTv}\\
	\nu  \left( \left| h_{\mathrm{sd}} \right| + \sum_{m = 1}^M | [ \mathbf{h}_{\mathrm{sr}} ]_m | | [\mathbf{h}_{\mathrm{rd}}]_m | \right)^2,&  \mbox{short-term}, \label{eq:OptSNRSTv}
\end{subnumcases}
where  the optimal long-term phase-shift matrix $\pmb{\Phi}^{\mathsf{opt}, \mathsf{lt}}$ is
\begin{equation} \label{eq:PhiLT}
	\pmb{\Phi}^{\mathsf{opt}, \mathsf{lt}} = \mathrm{diag}\big([e^{j\theta_{1}^{\mathsf{opt}, \mathsf{lt}}}, \ldots, e^{j\theta_{M}^{\mathsf{opt}, \mathsf{lt}}}]^T \big) \in \mathbb{C}^{M \times M},
\end{equation}
and  $ [ \mathbf{h}_{\mathrm{sr}} ]_m$ and $ [ \mathbf{h}_{\mathrm{rd}} ]_m$ are the $m$-th element of the cascaded channels $\mathbf{h}_{\mathrm{sr}}$ and $\mathbf{h}_{\mathrm{rd}}$, respectively.

The short-term phase shift design is, in general, an upper bound for the long-term phase shift design.
	In analytical terms, in fact, we have the following property
	\begin{equation} \label{eq:Ratebound}
		\begin{split}
		 & \log_2 (1 + \nu |h_{\mathrm{sd}} + \mathbf{h}_{\mathrm{sr}}^H \pmb{\Phi}^{\mathsf{opt}, \mathsf{lt}} \mathbf{h}_{\mathrm{rd}}|^2)   \stackrel{(a)}{\leq} \underset{\{ \theta_{m}^{\mathsf{st}} \}}{\max}\, \log_2 ( 1 + \gamma_n^{\mathsf{st}}) \\
   &\stackrel{(b)}{=} \log_2(1 + \nu |h_{\mathrm{sd}} + \mathbf{h}_{\mathrm{sr}}^H \pmb{\Phi}^{\mathsf{opt}, \mathsf{st}} \mathbf{h}_{\mathrm{rd}}|^2 ) ,
		\end{split}
	\end{equation}
where  $(a)$ and $(b)$ follow by definition of short-term and long-term phase shift optimization.\footnote{Regarding channel acquisition, channel statistics can be obtained by averaging many different channel measurements. Meanwhile, channel estimates can be obtained by a pilot training phase with a tolerable accuracy depending on the coherence time and the transmit power allocated to the pilot signals.}
\begin{remark}
The long-term and short-term phase shift designs are investigated in this paper based on either partial or full levels of CSI. To be specific,  channel statistics are only the information required to optimize the long-term phase shift coefficients that can be applied for multiple coherence intervals, where the large-scale fading information remains unchanged. We further compare this phase shift design with an upper bound, which relies on the instantaneous CSI.
\end{remark}

\section{Performance Evaluation}
In this section, we compute closed-form expressions for the coverage probability and ergodic rate by assuming long-term and short-term phase shift designs.
	\subsection{Coverage Probability}
Based on the rate in \eqref{eq:Rate}, the coverage probability for a given target rate $\xi$ [b/s/Hz]  is defined as 
	\begin{equation} \label{eq:PcovGen}
		P_{\mathsf{cov}} = 1 - \mathrm{Pr}(R < \xi), 
	\end{equation}
	where $\mathrm{Pr}(\cdot)$ denotes the probability of an event and $R$ is the rate based on either the long-term or short-term phase-shift designs, i.e., $R \in \{ R^{\mathrm{lt}}, R^{\mathrm{st}}  \}$. By denoting $z = \sigma^2 (2^\xi -1)$, the coverage probability can be formulated in terms of the SNR, as 
	\begin{equation} \label{eq:PcovRan}
		P_{\mathsf{cov}} =  1 - \mathsf{Pr}(\gamma < z),
	\end{equation}
	with $\gamma \in \{\gamma^{\mathrm{lt}}, \gamma^{\mathrm{st}} \}$.  Approximated closed-form expressions for \eqref{eq:PcovRan} that correspond to the long-term and short-term phase shift designs in \eqref{eq:Optphasse} and \eqref{eq:OptSt} are given in Theorem \ref{Theorem:CovProbRan} and Theorem \ref{Theorem:CovProbOpt}, respectively. The analytical frameworks in Theorem \ref{Theorem:CovProbRan} and Theorem \ref{Theorem:CovProbOpt} are obtained by utilizing the moment matching approach.\footnote{The generalization of the proposed approach to multi-user scenarios is not straightforward and may require additional assumptions to derive closed-form expressions, for example treating the multi-user interference as Gaussian noise. Under this approximation, the moment-matching method can be applied to derive closed-form expressions of the coverage probability and ergodic rate. This analysis is postponed to future work.} 
	\begin{theorem} \label{Theorem:CovProbRan}
		Consider the long-term phase-shift design. An approximated closed-form expression for the coverage probability in \eqref{eq:PcovRan} is given by
		\begin{equation} \label{eq:Pcovlt}
			P_{\mathsf{cov}}^{\mathrm{lt}} \approx \Gamma\left( k^{\mathrm{lt}}, z/w^{\mathrm{lt}} \right) /\Gamma (k^{\mathrm{lt}}),
		\end{equation}
		where the shape parameter $k^{\mathrm{lt}}$ and the scale parameter $w^{\mathrm{lt}}$ are
		\begin{align}
			k^{\mathrm{lt}} 
			&= \frac{ \left( 
				 \beta_{\mathrm{sd}} + 
				 o_{1}  +   \beta_{\mathrm{sr}}  \beta_{\mathrm{rd}}
				\right)^2 }{
				 \beta_{\mathrm{sd}}^2 +
				 o_{2} + 2 \beta_{\mathrm{sd}} o_{1} 
			}, \label{eq:klt}
			\\
			w^{\mathrm{lt}} 
			&=
			\frac{
				\nu \beta_{\mathrm{sd}}^2 +
				\nu o_{2} + 2 \nu \beta_{\mathrm{sd}} o_{1}
			}{
				 \beta_{\mathrm{sd}} + 
				 o_{1} \beta_{\mathrm{sr}} \beta_{\mathrm{rd}}
			}, \label{eq:wlt}
		\end{align}
		and $o_{1}$ and $o_{2}$ are defined as
		\begin{align}
			&o_{1} = \mu  (M^2 {\kappa_{\mathrm{sr}}}{ \kappa_{\mathrm{rd}}} + \tilde{\kappa} M ),\\
			&o_{2} = \mu^2 \left( M^2 \kappa_{\mathrm{sr}}\kappa_{\mathrm{rd}}( 2 {M\tilde{\kappa} + 8} ) +  M^2\tilde{\kappa}^2 + 2M\hat{\kappa} \right),
		\end{align}
	with $\mu$, $\tilde{\kappa}$, and $\hat{\kappa}$ given in Theorem~\ref{Theorem:RISChannel}.  
	\end{theorem}
	\begin{proof}
		See Appendix~\ref{Appendix:CovProbRan}.
	\end{proof}
Even though the SNR expressions in \eqref{eq:OptSNRLTv} and \eqref{eq:OptSNRSTv} do not follow a common distribution, we attain an approximated closed-form expression of the coverage probability by using the moment matching technique and by approximating the received SNR with a Gamma distribution. From the formulas, we observe the array gain provided by the RIS in the terms $o_{1}$ and $o_{2}$, which scale as $M^2$ and $M^3$ thanks to the optimal phase shift design, respectively. 
\begin{theorem} \label{Theorem:CovProbOpt}
		Consider the short-term phase shift design. An approximated closed-form expression for the coverage probability in \eqref{eq:PcovRan} is given by
		\begin{equation} \label{eq:Pcovst}
			P_{\mathsf{cov}}^{\mathrm{st}} \approx \Gamma\left( k^{\mathrm{st}}, z /w^{\mathrm{st}} \right) / \Gamma (k^{\mathrm{st}}),
		\end{equation}
		where the shape parameter $k^{\mathrm{st}}$ and the scale parameter $w^{\mathrm{st}}$ are
		\begin{align}
			k^{\mathrm{st}} = \frac{{k_{c}}\left( {{k_{c}} + 1} \right)}{{2 \left( {2{k_{c}} + 3} \right)}},
			w^{\mathrm{st}}  = {2{\nu}w_{c}^2\left( {2{k_{c}} + 3} \right) }.
		\end{align}
		with ${k_{c}} = \frac{{{{\left( {{c_{1}} + {c_{2}} } \right)}^2}}}{{{c_{3}} + {c_{4}}}}$, ${w_{c}} = \frac{{{c_{3}} + {c_{4}}}}{{{c_{1}} + {c_{2}}}}$, and $c_{1}, c_{2}, c_{3}, c_{4}$ are defined as
		\begin{align}
			& c_{1} = {\frac{1}{2}\sqrt {\pi {\beta_{\mathrm{sd}}}} },  c_{2} = {  \frac{\pi }{4}M \sqrt{\mu} t_{\mathrm{sr}}  t_{\mathrm{rd}}  }, \\
			& c_{3} = \frac{4 - \pi }{4}{\beta_{\mathrm{sd}}},  c_{4} = M \mu \left( {1 + \kappa_{\mathrm{sr}}} \right)\left( {1 + \kappa_{\mathrm{rd}}} \right) 
			- \frac{M{{\pi ^2}}}{{16}}\mu{ {t_{\mathrm{sr}}^2t_{\mathrm{rd}}^2}},
		\end{align}
and $t_{\mathrm{sr}} = {_1{F_1}\left( { - 0.5,1, - {\kappa_{\mathrm{sr}}}} \right)}$, $t_{\mathrm{rd}} =  {_1{F_1}\left( { - 0.5,1, - {\kappa_{\mathrm{rd}}}} \right)}$. 
	\end{theorem}
	\begin{proof}
		See Appendix~\ref{Appendix:CovProbOpt}.
	\end{proof}
The coverage probability in \eqref{eq:Pcovlt} for the long-term phase-shift design and in \eqref{eq:Pcovst} for the short-term phase-shift design provide simple but effective closed-form expressions for evaluating the performance of RIS-assisted communications without the need of resorting to Monte Carlo simulations or complex analytical frameworks. Even though the two approximate closed-form expressions of the coverage probability have a similar structure, the shape and scale parameters are different. For example, let us consider the case study when the number of RIS elements is usually sufficiently large. Then, the obtained analytical expressions can be further simplified. 
Specifically, ignoring the non-dominant terms in \eqref{eq:Pcovlt} and \eqref{eq:Pcovst}, the shape parameters tend to 
\begin{align} \label{eq:klstA}
&k^{\mathsf{lt}} \asymp \frac{o_{1}^2}{o_{2} + 2 \beta_{\mathrm{sd}}o_{1}},  k^{\mathsf{st}} \asymp \frac{k_{c}(k_{c} +1)}{2(2 k_{c} + 3)},
\end{align}
which grows with the number of RIS elements.
If $M \rightarrow \infty$, in addition, we obtain $ k^{\mathsf{lt}}, k^{\mathsf{st}}  \rightarrow \infty $.
Furthermore, the scale parameters tend to
\begin{equation} \label{eq:wlstA}
w^{\mathsf{lt}} \asymp \frac{
	\nu_{\mathrm{s}} o_{3} + 2 \nu_{\mathrm{s}} \beta_{\mathrm{s_1d}} o_{1}
}{o_{1} \beta_{\mathrm{s_1r}} \beta_{\mathrm{rd}}
},
 w^{\mathsf{st}} \asymp 2 \nu_{s}w_{C}^2 ( {2{k_{C}} + 3}).
\end{equation}
If $M \rightarrow \infty$, we obtain $w^{\mathsf{lt}}, w^{\mathsf{st}} \rightarrow \infty$. Thus, for both phase shift designs, the shape and scale parameters are unbounded from above, as the number of RIS elements goes large. By rewriting the upper incomplete Gamma function in a series expression, the coverage probability in \eqref{eq:Pcovlt} is simplified to 
\begin{equation} \label{eq:Pcovasymp}
P_{\mathsf{cov}} \stackrel{(a)}{=} 1 - \frac{(z/w^{\mathsf{x}})^{k^{\mathsf{x}}}}{k^{\mathsf{x}} \Gamma(k^{\mathsf{x}}) } \sum_{t=0}^{\infty} \frac{(-z/w^{\mathsf{x}})^t}{(k^{\mathsf{x}} + t)t!}  \stackrel{(b)}{\asymp} 1 - \frac{(z/w^{\mathsf{x}})^{k^\mathsf{x}}}{(k^{\mathsf{x}})^2 \Gamma(k^{\mathsf{x}})},
\end{equation}
where $\mathsf{x} \in \{ \mathsf{lt}, \mathsf{st} \}$. In \eqref{eq:Pcovasymp}, $(a)$ is obtained by utilizing the definition of upper incomplete Gamma function; $(b)$ follows because $1/w^{\mathsf{x}} \rightarrow 0$ as $M \rightarrow \infty$ and therefore the addends for $t \geq 1$ can be ignored. Based on the shape and scale parameters in \eqref{eq:klstA} and \eqref{eq:wlstA}, the approximation in \eqref{eq:Pcovasymp} unveils the coverage probability tends to $1$ as $M \rightarrow \infty$, for both the short-term and long-term phase shift designs. This implies that, if the number of RIS elements is large enough, an RIS is capable of offering good coverage.
	\subsection{Ergodic Rate}
	The channel rate in \eqref{eq:OptRate} with the optimal SNR in \eqref{eq:OptSNRLTv} and \eqref{eq:OptSNRSTv} depends on the small-scale and large-
	scale fading coefficients. Consequently, the channel rate by averaging out the small-scale fading is as follows:
	\begin{equation} \label{eq:ErgodicRate}
	\bar{R}_{\mathsf{x}} = \mathbb{E} \{ \log_2 ( 1 + \gamma^{\mathsf{x}} ) \}, \mbox{[b/s/Hz]}.
	\end{equation}
	  A closed-form expression for \eqref{eq:ErgodicRate} is given in Lemma~\ref{Lemma:ErgodicRate}, by using again the moment-marching approach and approximating the received SNR with a Gamma distribution. 
	\begin{lemma} \label{Lemma:ErgodicRate}
	The ergodic channel capacity in \eqref{eq:ErgodicRate} can be formulated in the closed-form expression as follows: 
	\begin{align} \label{eq:ClosedErgodicRate}
		\bar{R}^{\mathsf{x}} \approx \frac{1}{{\Gamma \left( {{k^{\mathsf{x}}}} \right)\ln \left( 2 \right)}}G_{2,3}^{3,1}\left( {\left. {\frac{1}{{{w^{\mathsf{x}}}}}} \right|\begin{array}{*{20}{c}}
				{0,1}\\
				{0,0,{k^{\mathsf{x}}}}
		\end{array}} \right),
	\end{align}
where $G_{p,q}^{m,n}\Big( { z \Big|\begin{array}{*{20}{c}}
		{a_1,\ldots, a_q}\\
		{b_1,\ldots,b_p}
\end{array}} \Big)$ is  the  Meijer-G  function, and, similar to Theorems~\ref{Theorem:CovProbRan} and \ref{Theorem:CovProbOpt}, $k^{\mathsf{x}}$ and $w^{\mathsf{x}}$ are the shape and scale parameters of the approximating Gamma distribution, respectively.
	\end{lemma}
	\begin{IEEEproof}
	The proof is similar to the proof in \cite{9195523}, but it is adapted to the channel model considered in this paper.
	\end{IEEEproof}
Differently from \cite{9195523}, Lemma~\ref{Lemma:ErgodicRate} can be applied to all phase shift designs, which include the short-term and long-term phase shift designs of interest in this paper. 
The analytical expressions for $k$ and $w$ that correspond to the latter phase shift designs are the same as in Theorem~\ref{Theorem:CovProbRan}.

	\section{RIS Placement Optimization}
From the obtained closed-form expression of the coverage probability, this section formulates and solves the coverage probability maximization problem as a function of the location of the RIS, based on the Lagrangian function and the gradient ascent method.\footnote{A similar approach may be utilized to optimize the ergodic rate as a function of the position of the RIS.}
	\subsection{Problem Formulation}
Besides the optimization of the phase shifts of the RIS, it is possible, in many cases, to optimize the location of the RIS as well. This may be case of cellular networks, where the locations of the RIS may be optimized \cite{sihlbom2022reconfigurable}, or when the RIS may be placed on a moving object \cite{geraci2022will}. For example, the RIS location can be optimized via a specific utility metric to enhance the spectral and energy efficiency or the coverage. To this end, we denote the RIS coordinates as $(x_{\mathrm{r}}, y_{\mathrm{r}}, z_{\mathrm{r}})$, and we impose some constraints on the feasible locations of the RIS
	\begin{equation}
		x_{\mathrm{r},\min} \leq x_{\mathrm{r}}\leq x_{\mathrm{r},\max},  y_{\mathrm{r},\min} \leq y_{\mathrm{r}}\leq y_{\mathrm{r},\max}, z_{\mathrm{r},\min} \leq z_{\mathrm{r}}\leq z_{\mathrm{r},\max},
	\end{equation}
	where $(x_{\mathrm{r},\min}, y_{\mathrm{r},\min}, z_{\mathrm{r},\min})$ and $(x_{\mathrm{r},\max}, y_{\mathrm{r},\max}, z_{\mathrm{r},\max})$ denote the range of locations in which the RIS can be deployed. In an outdoor environment, for example, the RIS may be located only on buildings or even only on a small portion of the available buildings. Given the RIS location, the distances from  the source to the destination through the RIS are 
	\begin{align}
		d_{\mathrm{sr}} &= \sqrt{(x_{\mathrm{r}} - x _{\mathrm{s}})^2 + (y_{\mathrm{r}} - y _{\mathrm{s}})^2 + (z_{\mathrm{r}} - z _{\mathrm{s}})^2},\\
		d_{\mathrm{rd}} &= \sqrt{(x_{\mathrm{r}} - x _{\mathrm{d}})^2 + (y_{\mathrm{r}} - y _{\mathrm{d}})^2 + (z_{\mathrm{r}} - z _{\mathrm{d}})^2}, 
	\end{align}
	where $(x_\mathrm{s}, y_\mathrm{s}, z_\mathrm{s})$ and $(x_\mathrm{d}, y_\mathrm{d}, z_\mathrm{d})$ are the locations of the source and the destination, respectively. Also, the large-scale fading coefficients can be formulated in terms of the propagation distances as
	\begin{equation} \label{eq:betaFormulation}
		\beta_{\mathrm{sr}} = K_0 d_{\mathrm{sr}}^{-\eta_{\mathrm{sr}}} \mbox{ and } \beta_{\mathrm{rd}} = K_0 d_{\mathrm{rd}}^{-\eta_{\mathrm{rd}}},
	\end{equation}
	where $K_0$ is the path loss in [dB] at the reference distance of $1$~m; and $\eta_{\mathrm{sr}}$ and $\eta_{\mathrm{rd}}$ are the path loss exponents. The RIS placement optimization problem can be formulated  as follows:
	\begin{equation}  \label{Prob:UAV}
		\begin{aligned}
			&\underset{\{ x_{\mathrm{r}}, y_{\mathrm{r}}, z_{\mathrm{r}} \}}{\mathrm{maximize}} &\quad & P_{\mathsf{cov}}  ( x_{\mathrm{r}}, y_{\mathrm{r}}, z_{\mathrm{r}} )  \\
			&\mbox{subject to} &&  x_{\mathrm{r},\min} \leq   x_{\mathrm{r}} \leq x_{\mathrm{r},\max},\\
			&&&  y_{\mathrm{r},\min} \leq   y_{\mathrm{r}} \leq y_{\mathrm{r},\max}, \\
			&&& z_{\mathrm{r},\min} \leq   z_{\mathrm{r}} \leq z_{\mathrm{r},\max}, 
		\end{aligned}
	\end{equation}
which maximizes the coverage probability  by considering the RIS location $(x_{\mathrm{r}}, y_{\mathrm{r}}, z_{\mathrm{r}})$ as the optimization variables. Since the optimization problem in \eqref{Prob:UAV} is non-convex, the global optimum is extremely difficult to obtain in general even though it may be found through an exhaustive search \cite{sihlbom2022reconfigurable}.\footnote{Even though global optimization toolboxes may be available, they cannot, in general, guarantee to obtain the global optimum of the considered coverage probability maximization problem due to the non-convexity of the objective function and mathematical difficulty of computing the first-order and second-order derivative of the partial Lagrangian function with respect to the location of the RIS. In addition, these methods entail a prohibitive computational complexity, in general.}

\subsection{Solution to the RIS Placement Optimization Problem}
For the sake of compactness, we use the notation $\varrho_{\mathrm{r}} \in \mathcal{R}$ with $ \mathcal{R} = \{ x_{\mathrm{r}}, y_{\mathrm{r}}, z_{\mathrm{r}}\}$. We propose to obtain a low-complexity solution by formulating the partial Lagrangian function of the problem in \eqref{Prob:UAV} as
\begin{equation}
\mathcal{L}( \{ \varrho_{\mathrm{r}} \} ) =   P_{\mathsf{cov}} (  \{\varrho_{\mathrm{r}} \} ).
\end{equation}
In general, the Rician factors $K_{\mathrm{s_1r}}$ and $K_{\mathrm{rd}}$ depend on the locations of the RIS as well, i.e., the link distances. However, their variations are usually smaller as compared with the large-scale fading coefficients \cite{3gpp2011technical}. For simplicity, therefore, we assume that they are independent of the location of the RIS. Under this assumption, an approximated expression of the first-order derivative of the Lagrangian function is given in Theorem~\ref{Theorem:1stDerivative}.

\begin{theorem} \label{Theorem:1stDerivative}
If the Rician factors are assumed  approximately independent of the RIS location, the first-order derivative of  $\mathcal{L}(\{\varrho_{\mathrm{r}} \})$ with respect to $\varrho_{\mathrm{r}}$ is 
\begin{equation} \label{eq:FirstDerivative}
\dt{\mathcal{L}} ( \{\varrho_{\mathrm{r}} \}) =  \frac{\partial \mathcal{L}( \{ \varrho_{\mathrm{r}} \} )  }{\partial \varrho_{\mathrm{r}}} \approx \dt{P}_{\mathsf{cov}}  (\{\varrho_{\mathrm{r}}\}),
\end{equation}
where $\dt{P}_{\mathsf{cov}}  ( \{\varrho_{\mathrm{r}} \}) $ is the first-order derivative of the coverage probability $P_{\mathsf{cov}} (  \{\varrho_{\mathrm{r}} \} ) $ with respect to $\varrho_{\mathrm{r}}$ and conditioned on the Rician factors. In particular, $\dt{P}_{\mathsf{cov}}  ( \{\varrho_{\mathrm{r}} \})$ can be formulated as
\begin{equation} \label{eq:1stPcov}
\dt{P}_{\mathsf{cov}} ( \varrho_{\mathrm{r}})  =  \frac{{ \dt{N}\left( \varrho_{\mathrm{r}} \right)D\left( \varrho_{\mathrm{r}} \right) - N\left( \varrho_{\mathrm{r}} \right) \dt{D}\left( \varrho_{\mathrm{r}} \right)}}{{{{ D^2 \left( \varrho_{\mathrm{r}} \right) }}}},
\end{equation}
where $N\left( \varrho_{\mathrm{r}} \right)$ and $D\left( \varrho_{\mathrm{r}} \right)$ are given by
\begin{align} \label{eq:NDDef}
&N\left( \varrho_{\mathrm{r}} \right) = 	\begin{cases}
\Gamma(k^{\mathsf{st}}, z/w^{\mathsf{st}}), & \mbox{short-term},\\
\Gamma(k^{\mathsf{lt}}, z/w^{\mathsf{lt}}), & \mbox{long-term},
\end{cases},\\
& D\left( \varrho_{\mathrm{r}} \right) = 	\begin{cases}
	\Gamma(k^{\mathsf{st}}), & \mbox{short-term},\\
	\Gamma(k^{\mathsf{lt}}), & \mbox{long-term},
\end{cases}
\end{align}
and $\dt{N}\left( \varrho_{\mathrm{r}} \right)$ and $\dt{D}\left( \varrho_{\mathrm{r}}\right)$ are the first-order derivative of $N\left( \varrho_{\mathrm{r}} \right)$ and $D\left( \varrho_{\mathrm{r}} \right)$, respectively, which are defined in Appendices~\ref{Appendix:1stPcov} and \ref{Appendix:1stPcovv1} for the short-term and long-term phase shift design, respectively.
\end{theorem}
Given an initial value $x_{\mathrm{r}}^{(0)}, y_{\mathrm{r}}^{(0)},$ and $z_{\mathrm{r}}^{(0)}$ in the feasible set of locations of the RIS, one can find a good sub-optimal solution to problem~\eqref{Prob:UAV} by utilizing the gradient ascent method. In the $n$-th iteration, the location of the RIS is updated as
\begin{equation} \label{eq:varrhoupdated}
\varrho_{\mathrm{r}}^{(n)} \leftarrow \varrho_{\mathrm{r}}^{(n-1)} + \mu \dt{\mathcal{L}} ( \{\varrho_{\mathrm{r}} \})\big|_{\varrho_{\mathrm{r}} = \varrho_{\mathrm{r}}^{(n-1)}}, \forall \varrho_{\mathrm{r}} \in \mathcal{R},
\end{equation}
where the step size $\mu > 0$ is sufficiently large in the direction of the steepest ascent. Because of the constraints on feasible set of locations of the RIS in \eqref{Prob:UAV}, the location of the RIS is updated by checking the limiting values as follows:
\begin{equation} \label{eq:varhorn}
\varrho_{\mathrm{r}}^{(n)} \leftarrow \min \left( \max \left( \varrho_{\mathrm{r}}^{(n)}, \varrho_{\mathrm{r}, \min}\right), \varrho_{\mathrm{r}, \max} \right), \forall \varrho_{\mathrm{r}} \in \mathcal{R},
\end{equation}
where $\varrho_{\mathrm{r}, \min} \in \mathcal{R}_{\min} = \{x_{\mathrm{r}, \min}, y_{\mathrm{r}, \min}, z_{\mathrm{r}, \min}  \}$ and $\varrho_{\mathrm{r}, \max}  \in \mathcal{R}_{\max} = \{x_{\mathrm{r}, \max}, y_{\mathrm{r}, \max}, z_{\mathrm{r}, \max}  \}$. Recalling that the approximation in \eqref{eq:FirstDerivative} is derived by assuming that the Rician factors are independent of the location at the $n$-th iteration, they can be updated as
\begin{equation} \label{eq:Kappan}
\kappa_{\alpha}^{(n)} = f \left(x_{\mathrm{r}}^{(n)}, y_{\mathrm{r}}^{(n)}, z_{\mathrm{r}}^{(n)} \right).
\end{equation}
The update in \eqref{eq:varhorn} is repeated until the RIS location converges to a sub-optimal solution. The convergence criterion of the proposed RIS placement algorithm is based on evaluating the small variation of two consecutive iterations as
\begin{equation} \label{eq:Stopping}
\left| x_{\mathrm{r}}^{(n)} - x_{\mathrm{r}}^{(n-1)}  \right|^2  + \left| y_{\mathrm{r}}^{(n)} - y_{\mathrm{r}}^{(n-1)}  \right|^2 + \left| z_{\mathrm{r}}^{(n)} - z_{\mathrm{r}}^{(n-1)}  \right|^2 \leq \epsilon,
\end{equation}
where $\epsilon >0$ is a given accuracy constant. Once the convergence is achieved at the $n$-th iteration, a  sub-optimal solution is given by
\begin{equation} \label{eq:Suboptimal}
x_{\mathrm{r}}^{\ast} = x_{\mathrm{r}}^{(n)}, x_{\mathrm{r}}^{\ast} = x_{\mathrm{r}}^{(n)}, \mbox{ and } z_{\mathrm{r}}^{\ast} = z_{\mathrm{r}}^{(n)},
\end{equation}
with the corresponding coverage probability $P_{\mathsf{cov}}(\{ \varrho_{\mathrm{r}}^{\ast} \})$  and $\varrho_{\mathrm{r}}^{\ast} \in \mathcal{R}^{\ast} = \{ x_{\mathrm{r}}^\ast, y_{\mathrm{r}}^\ast, z_{\mathrm{r}}^\ast \}$. The proposed iterative approach is summarized in Algorithm~\ref{Algorithm1}.\footnote{The convergence of Algorithm~\ref{Algorithm1} is established based on the methodology inherent in the gradient ascent method \cite{nocedal1999numerical}.} As far as the computational complexity is concerned, the highest cost of each iteration of Algorithm~\ref{Algorithm1} is the evaluation of the first-order derivative $\dt{\mathcal{L}} ( \{\varrho_{\mathrm{r}} \})$, denoted by $f$. Then, the computational complexity of Algorithm~\ref{Algorithm1} is in the order of $\mathcal{O}(Nf)$, where $N$ is the number of iterations needed to reach a fixed-point solution. By utilizing numerical evaluations, we observe that $N$ is in an order of hundreds of iterations. Instead of Algorithm~\ref{Algorithm1}, we may use a grid-search approach. If we denote the number of points of the grid as $N_x, N_y,$ and $N_z$ for the three dimensions, the computational complexity of a grid-based exhaustive search is in the order of $\mathcal{O}(N_x N_y N_z \tilde{f})$, where $\tilde{f}$ is the cost for evaluating the coverage probability. For example, setting $N_x=N_y=N_z= 100$, the exhaustive search has a computational complexity in the order of $10^6$, which is impractical. Hence, Algorithm~\ref{Algorithm1} has much lower computational complexity and is more practical than a grid-based exhaustive search approach.
\begin{algorithm}[t]
	\caption{Sub-optimal solution to problem~\eqref{Prob:UAV} by utilizing the gradient ascent method} \label{Algorithm1}
	\textbf{Input}:  The reference path loss $K_0$; the source coordinate $(x_\mathrm{s}, y_\mathrm{s}, z_\mathrm{s})$; the destination  coordinate $(x_\mathrm{d}, y_\mathrm{d}, z_\mathrm{d})$; the initial coordinate of the RIS $(x_\mathrm{r}^{(0)}, y_\mathrm{r}^{(0)}, z_\mathrm{r}^{(0)})$; the large-scale fading coefficients $\beta_{\mathrm{sd}}$, $\beta_{\mathrm{sr}}^{(0)}$, $\beta_{\mathrm{rd}}^{(0)}$; the transmit power $p$; and the noise variance $\sigma^2$.
	\begin{itemize}
		\item[1.] Set $n=0$ and the initial values $\{ \varrho_{\mathrm{r}}^{(0)} \} = \big\{ x_{\mathrm{r}}^{(0)}, y_{\mathrm{r}}^{(0)}, z_{\mathrm{r}}^{(0)}  \big\}$.
		\item[2.] Compute the first-order derivative of $\mathcal{L}( \{ \varrho_{\mathrm{r}} \} )$ with respect to $\varrho_{\mathrm{r}}$ at the coordinates $x_{\mathrm{r}}^{(n)}$, $y_{\mathrm{r}}^{(n)}$, and $z_{\mathrm{r}}^{(n)}$ by  \eqref{eq:FirstDerivative}.
		\item[3.] Update the coordinates  $\{\varrho_{\mathrm{r}}^{(n)} \}$ of the RIS by the gradient ascent method in \eqref{eq:varrhoupdated}.
		\item[4.] Update the coordinates $\{\varrho_{\mathrm{r}}^{(n)} \}$ of the RIS by checking the boundary conditions in \eqref{eq:varhorn}.
		\item[5.] Update the Rician factors by utilizing \eqref{eq:Kappan}.
		\item[6.] Verify the convergence criterion: If \eqref{eq:Stopping} is satisfied $\rightarrow$ set the sub-optimal solution as in \eqref{eq:Suboptimal}. Otherwise, set $n=n+1$ and repeat steps $2-4$.
	\end{itemize}
	\textbf{Output}: The location of the RIS $(x_{\mathrm{r}}^\ast, y_{\mathrm{r}}^\ast, z_{\mathrm{r}}^\ast )$.
\end{algorithm}

	\begin{figure}[t]
		\centering
		\includegraphics[trim=0.9cm 0.0cm 0.9cm 0.6cm, clip=true, width=3.2in]{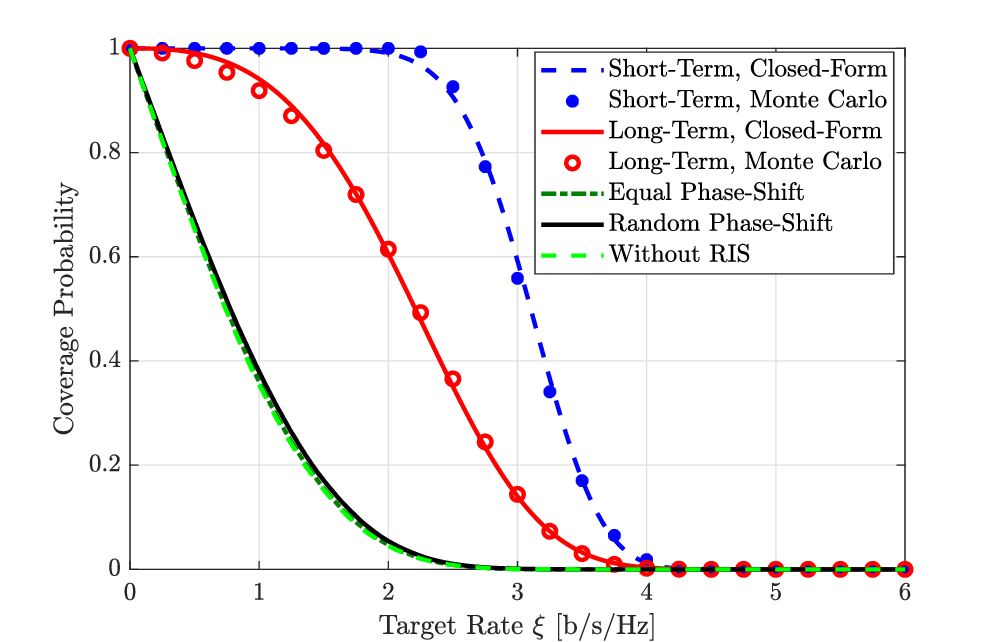} \vspace*{-0.25cm}
		\caption{Coverage probability versus the target rate [b/s/Hz] with $M=64$, $(x_{\mathrm{r}},  y_{\mathrm{r}}, z_{\mathrm{r}} ) = (27, 25, 25)$~[m], $(x_{\mathrm{d}},  y_{\mathrm{d}}, z_{\mathrm{d}} ) = (180, 100, 25)$~[m], and $p = 20$~[dBm].}
		\label{FigMonteCarloClosedForm}
		\vspace*{-0.0cm}
	\end{figure}
	\begin{figure}[t]
		\centering
		\includegraphics[trim=0.9cm 0cm 0.9cm 0.6cm, clip=true, width=3.2in]{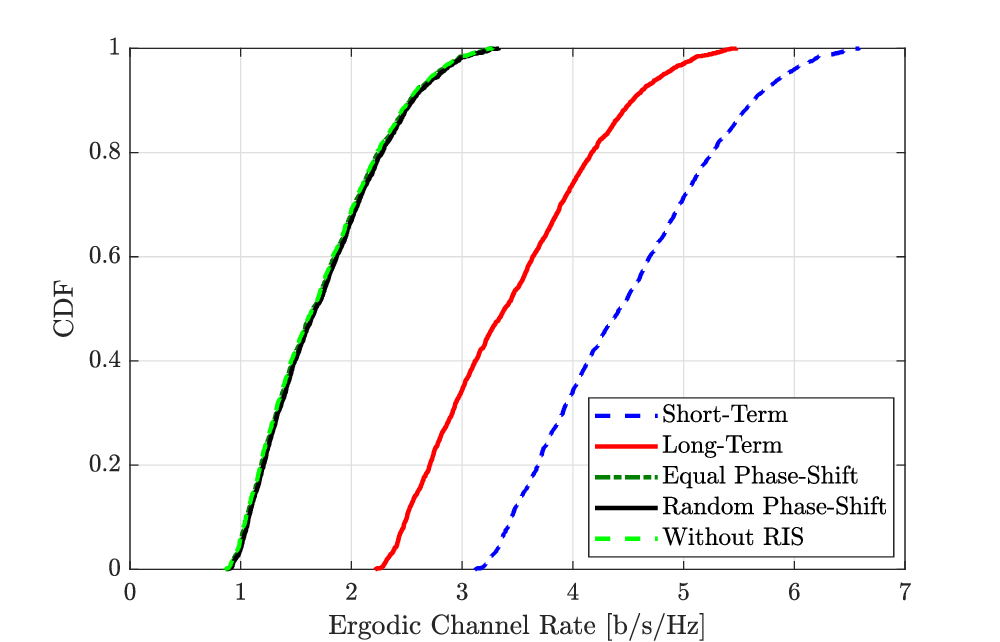} \vspace*{-0.25cm}
		\caption{Cumulative distribution function of the ergodic channel rate [b/s/Hz] with $M=64$, $(x_{\mathrm{r}},  y_{\mathrm{r}}, z_{\mathrm{r}} ) = (27, 25, 25)$~[m], $p = 20$~[dBm], and the destination location: $100$~[m] $\leq x_{\mathrm{d}} \leq $ $180$~[m],  $50$~[m] $\leq x_{\mathrm{d}} \leq $ $100$~[m], $z_{\mathrm{d}} = 15$~[m]. }
		\label{FigCDFRate}
		\vspace*{-0.2cm}
	\end{figure}

\begin{figure}[t]
\begin{minipage}{0.48\textwidth}
		\centering
		\includegraphics[trim=0.9cm 0cm 0.5cm 0.6cm, clip=true, width=3.4in]{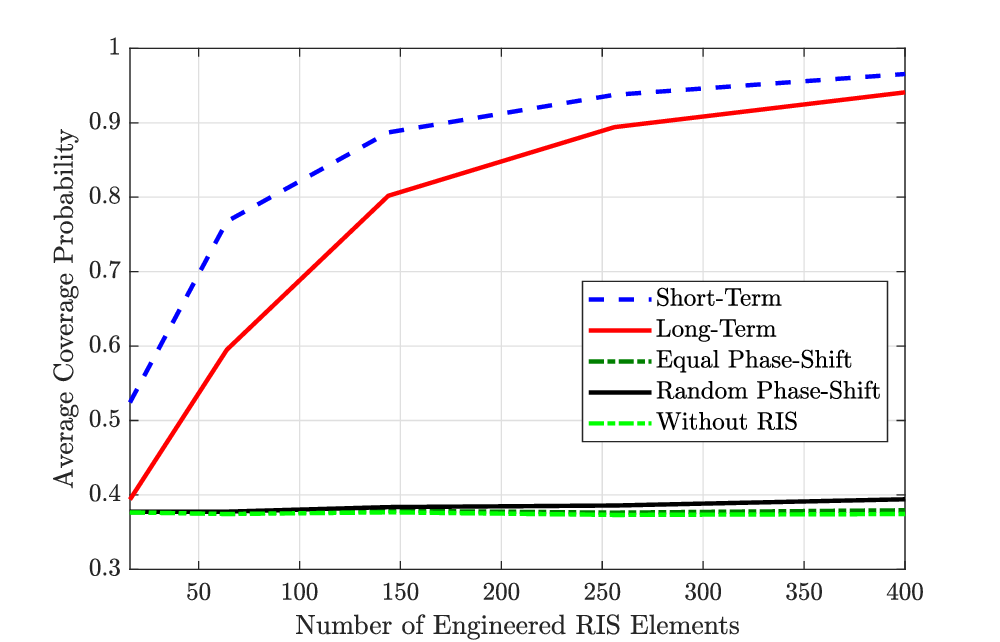} \vspace*{-0.4cm}
        \\ (a)
		\vspace*{-0.0cm}
\end{minipage}
	\vspace*{-0.0cm}
	\hfill
\begin{minipage}{0.48\textwidth}
		\centering
		\includegraphics[trim=0.9cm 0cm 0.9cm 0.6cm, clip=true, width=3.3in]{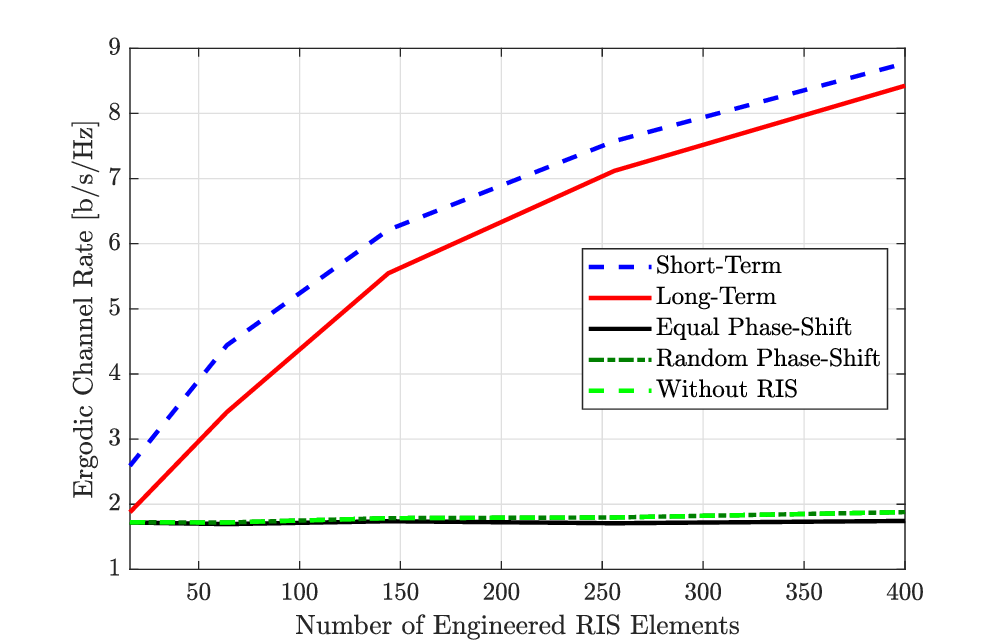} \vspace*{-0.2cm} \\ (b)
		\vspace*{-0.0cm}
  \end{minipage}
  \caption{System performance versus the number of RIS elements with $(x_{\mathrm{r}},  y_{\mathrm{r}}, z_{\mathrm{r}} ) = (27, 25, 25)$~[m], $p = 20$~[dBm], and the destination location: $100$~[m] $\leq x_{\mathrm{d}} \leq $ $180$~[m],  $50$~[m] $\leq y_{\mathrm{d}} \leq $ $100$~[m], $z_{\mathrm{d}} = 15$~[m]: $(a)$ The average coverage probability and $(b)$ the ergodic channel rate. } \label{FigWithWithoutCooperationST}
\end{figure}
\begin{figure}[t]
\begin{minipage}{0.48\textwidth}
		\centering
		\includegraphics[trim=0.5cm 0cm 0.5cm 0.6cm, clip=true, width=3.3in]{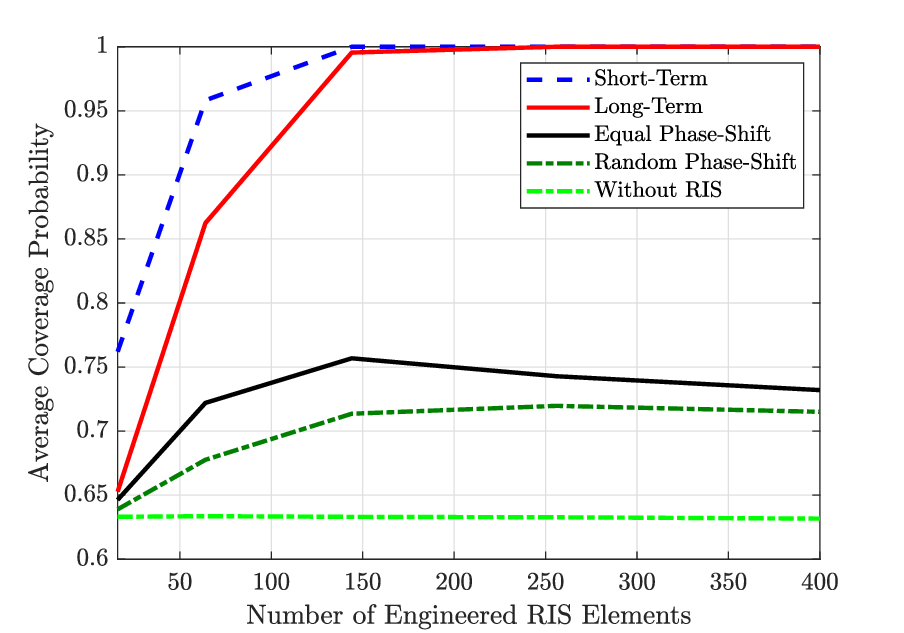} \vspace*{-0.1cm}
        \\ (a)
		\vspace*{-0.0cm}
\end{minipage}
	\vspace*{-0.0cm}
	\hfill
\begin{minipage}{0.48\textwidth}
		\centering
		\includegraphics[trim=0.5cm 0cm 0.9cm 0.6cm, clip=true, width=3.2in]{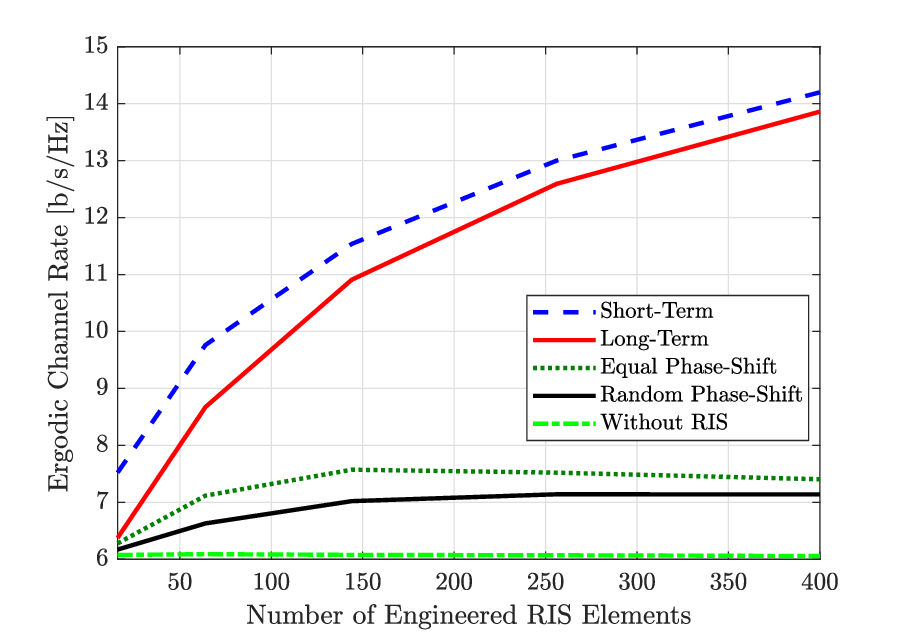} \vspace*{-0.1cm} \\ (b)
		\vspace*{-0.0cm}
  \end{minipage}
  \caption{System performance versus the number of RIS elements with $(x_{\mathrm{r}},  y_{\mathrm{r}}, z_{\mathrm{r}} ) = (27, 25, 25)$~[m], $p = 20$~[dBm], and the destination location: $40$~[m] $\leq x_{\mathrm{d}} \leq $ $55$~[m],  $10$~[m] $\leq y_{\mathrm{d}} \leq $ $25$~[m], $z_{\mathrm{d}} = 15$~[m]: $(a)$ The average coverage probability and $(b)$ the ergodic channel rate.} \label{FigWithWithoutCooperationSTv1}
\end{figure}
\section{Numerical Results}
In this section, the proposed analytical and optimization frameworks are validated by Monte Carlo simulations. The source is located at the origin. The direct link between the source and the destination is assumed to be weak and the channel gain  $\beta_{\mathrm{sd}}$~[dB] is $\beta_{\mathrm{sd}} = -33.1 -3.5 \log_{10} (d_{\mathrm{sd}}/1~\mathrm{m})$. For the indirect link, the channel gains $\beta_{\alpha}$~[dB] are 
$  \beta_{\alpha} = -25.5 -2.4\log_{10} (d_{\alpha}/1~\mathrm{m}),$
where $d_{\alpha}$ is the distance between the source and the receiver, i.e., the RIS if $\alpha = \mathrm{sr}$ and the destination if $\alpha = \mathrm{rd}$. The Rician factors are equal to $\kappa_\alpha = 10^{1.3 - 0.003 d_{\alpha}}$ \cite{3gpp2011technical}. The transmit power is $20$~mW, and the system bandwidth is $20$~MHz. The carrier frequency is $1.8$~GHz, and the noise power is $-94$~dBm, corresponding to a noise figure of $10$~dB. We consider the following phase shift designs for comparison:
\begin{itemize}
	\item[$i)$] The long-term phase shift design optimizes the average received SNR as formulated in \eqref{eq:Optphasse}. This phase shift design is denoted as ``Long-term" in the figures.
	\item[$ii)$] The short-term phase shift design optimizes the instantaneous SNR as formulated in \eqref{eq:OptSt}.  This phase shift design is denoted as ``Short-term" in the figures.
	\item[$iii)$] The equal phase shift design where the phase shifts are all equal to each other. This phase shift design is denoted as ``Equal phase" in the figures. 
	\item[$iv)$] The random phase shift design where arbitrary values in the range $[-\pi, \pi]$ are considered. This phase shift design is denoted as ``Random phase" in the figures. 
\end{itemize}
\begin{figure}[t]
		\centering
		\includegraphics[trim=0.9cm 0cm 0.5cm 0.6cm, clip=true, width=3.5in]{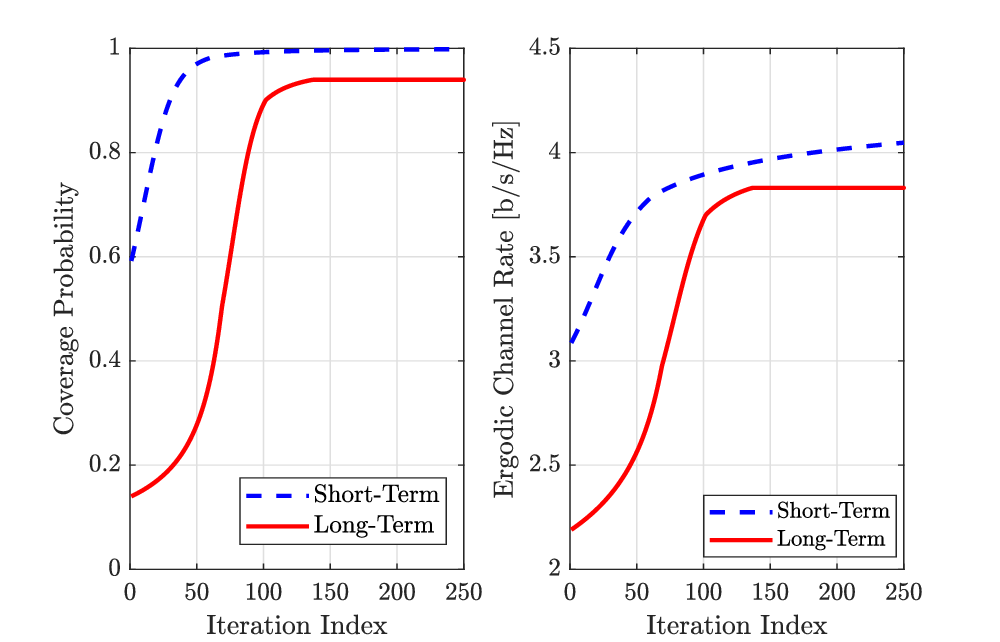} \vspace*{-0.25cm}
		\caption{Coverage probability and ergodic channel rate [b/s/Hz] with respect to the number of iterations  of Algorithm~\ref{Algorithm1} with $M=64$, $p = 20$~[dBm], and $(x_{\mathrm{d}}, y_{\mathrm{d}}, z_{\mathrm{d}}) = (180, 100, 15)$ [m]. The optimal RIS position is $(x_r^\ast, y_r^\ast, z_r^\ast) = (20, 10, 5)$ [m] and $(x_r^\ast, y_r^\ast, z_r^\ast) = (20, 13.36, 11.81)$ [m] corresponding the long-term and short-term phase shift design, respectively.}
		\label{Fig64dB20dBm}
		\vspace*{-0.2cm}
\end{figure}
\begin{figure}[t]
		\centering
		\includegraphics[trim=0.9cm 0cm 0.9cm 0.6cm, clip=true, width=3.5in]{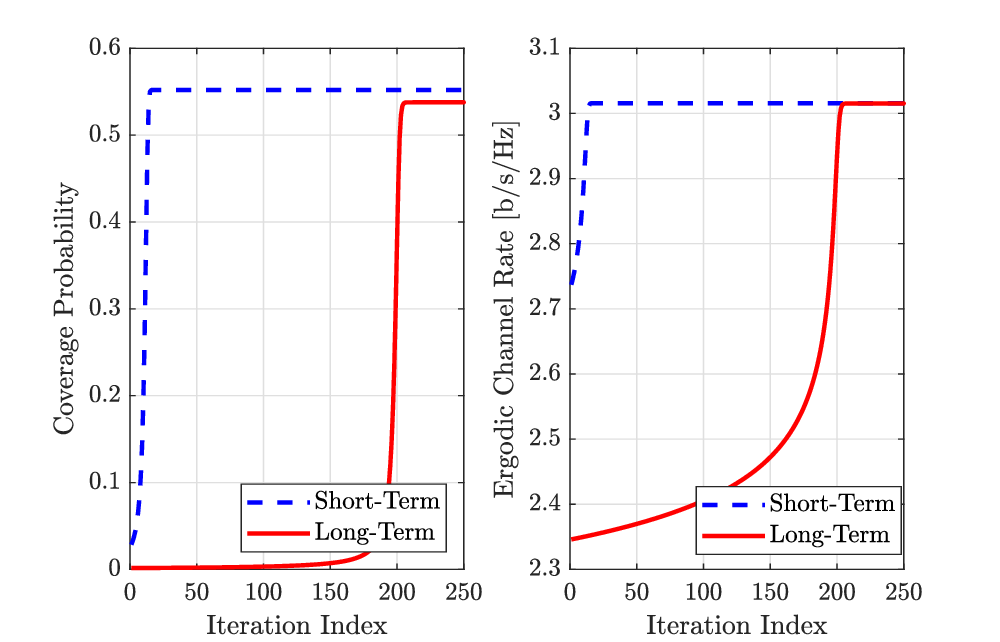} \vspace*{-0.25cm}
		\caption{Coverage probability and ergodic channel rate [b/s/Hz] with respect to the number of iterations of Algorithm~\ref{Algorithm1} with $M=256$, $p = 9$~[dBm], and $(x_{\mathrm{d}}, y_{\mathrm{d}}, z_{\mathrm{d}}) = (180, 100, 15)$ [m]. The optimal RIS position is $(x_r^\ast, y_r^\ast, z_r^\ast) = (21.56, 19.03, 17.84)$ [m] and $(x_r^\ast, y_r^\ast, z_r^\ast) = (24.50, 22.13, 21.41)$ [m] corresponding the long-term and short-term phase shift design, respectively. }
		\label{Fig256E9dBm}
		\vspace*{-0.2cm}
\end{figure}

In Fig.~\ref{FigMonteCarloClosedForm}, we compare the closed-form expression of the coverage probability by using \eqref{eq:Pcovlt} in Theorem~\ref{Theorem:CovProbRan} and \eqref{eq:Pcovst} in Theorem~\ref{Theorem:CovProbOpt} against Monte Carlo simulations by using \eqref{eq:PcovGen}. The good match between the analytical results and the numerical simulations confirms the accuracy of the proposed analytical framework. The proposed phase shift designs offer a coverage probability significantly higher than the considered benchmark schemes for a given set of instantaneous channel coefficients. For instance, the random and equal phase shift designs give a coverage probability of about $0.2$ if the target rate is 2 [b/s/Hz]. On the other hand, the long-term phase shift design has a coverage probability of about $0.6$ for the considered target rate, which is $3\times$ better than the baselines. By exploiting instantaneous CSI, the short-term phase shift design  can offer full coverage probability at the target rate of $2$~[b/s/Hz]. 

The cumulative distribution function (CDF) of the ergodic rate for different user locations is shown in Fig.~\ref{FigCDFRate}. The random phase shift design and a system without the presence of the RIS provide an ergodic channel rate of $1.75$~[b/s/Hz] on average. The long-term phase shift design increases the ergodic rate by a factor $1.97\times$ with respect to the two baselines, which is $3.46$~[b/s/Hz]. Meanwhile, the short-term phase shift design gives the best performance with a rate of $4.49$~[b/s/Hz], which is $1.3\times$ higher than the long-term phase shift design. Assuming a coverage target of $0.95$, the long-term and short-term phase-shift designs offer a $2.49\times$ and $3.42\times$ gain compared to the benchmark schemes, respectively. 

In Fig.~\ref{FigWithWithoutCooperationST}$(a)$, we utilize the analytical framework in Theorems~\ref{Theorem:CovProbRan} and \ref{Theorem:CovProbOpt} to evaluate the coverage probability as a function of the number of
RIS elements and for different designs of the phase shifts.
The phase shift designs  based on short-term and long-term CSI offer significant gains compared to the equal and random phase shift designs. The gains become significantly larger by increasing the number of RIS elements. Specifically, the short-term phase shift design provides a   $1.37\times$ gain with respect to  the random phase-shift design if the RIS is equipped with $16$ RIS elements. The gain is $2.49\times$ for $400$ RIS elements.  In addition, the gap between the
short-term and long-term phase shifts design reduces as the number of RIS elements increases. In Fig.~\ref{FigWithWithoutCooperationST}$(b)$, we display the ergodic rate [b/s/Hz] in \eqref{eq:ClosedErgodicRate}. We evince that the deployment of an RIS results in a substantial increase of the ergodic rate, as opposed to surfaces that operate as random scatterers and are not smart and reconfigurable. Notably, the long-term phase shift design provides an ergodic rate close to the short-term phase shift design and approaches it if the number of RIS elements is sufficiently large. If the RIS elements increase from  $16$ to $400$, the gap between the short-term and long-term phase shift designs reduces from $1.38\times$ to $1.04\times$.

In Fig.~\ref{FigWithWithoutCooperationSTv1}, the system performance is evaluated when the location of the destination is uniformly distributed in a narrower range compared with the setting considered in the previous figures. In this case, we see that the random and equal phase-shift designs outperform the case study without RIS. Furthermore, the results in Fig.~\ref{FigWithWithoutCooperationSTv1} confirm that the phase shift design based on statistical CSI provides system performance close to the optimal solution obtained by exploiting instantaneous CSI, especially when the large number of RIS elements is large.

In Figs.~\ref{Fig64dB20dBm} and \ref{Fig256E9dBm}, we illustrate the coverage probability as a function of the iterations of Algorithm~\ref{Algorithm1}. The initial location of the RIS $\big(x_{\mathrm{r}}^{(0)}, y_{\mathrm{r}}^{(0)}, z_{\mathrm{r}}^{(0)} \big)$ is selected as in Figs.\ref{FigMonteCarloClosedForm}--\ref{FigWithWithoutCooperationST}, $(x_{\mathrm{r}, \min}, y_{\mathrm{r}, \min}, z_{\mathrm{r}, \min}) = (20, 10, 5)$ [m], $(x_{\mathrm{r}, \max}, y_{\mathrm{r}, \max}, z_{\mathrm{r}, \max}) = (30, 40, 35)$ [m], and the step size is $\mu = 0.9$. We observe a remarkable enhancement of the coverage probability by appropriately deploying the RIS. In Fig.~\ref{Fig64dB20dBm}, assuming $M=64$, $p=20$~[dBm], the short-term phase shift design improves the coverage probability from $0.59$ if the RIS is located in $(x_{\mathrm{r}}^{(0)}, y_{\mathrm{r}}^{(0)}, z_{\mathrm{r}}^{(0)}) = (27, 25, 25)$~[m] to $1.0$ if the RIS is deployed in the location returned by Algorithm~\ref{Algorithm1}, i.e.,  $(x_{\mathrm{r}}^{\ast}, y_{\mathrm{r}}^{\ast}, z_{\mathrm{r}}^{\ast}) = (20, 13.36, 11.81)$. It corresponds to a $1.69\times$ improvement of the coverage probability. For the long-term phase shift design, Algorithm~\ref{Algorithm1} returns an RIS location at the coordinates $(x_{\mathrm{r}}^{\ast}, y_{\mathrm{r}}^{\ast}, z_{\mathrm{r}}^{\ast}) = (20, 10, 5)$ and offers a $6.71 \times$ improvement in the coverage probability compared with the initial point. 
Besides, Algorithm~\ref{Algorithm1} applied to the short-term phase shift design converges faster as compared to when it is applied to the long-term phase shift design, especially if the number of RIS elements is large. Interestingly, the results demonstrate that maximizing the coverage probability with respect to the optimal location significantly boosts the ergodic channel rate as well. Compared to Fig.~\ref{Fig64dB20dBm}, the number of RIS elements increases by $4\times$ and the transmit power reduces by approximately $2.2 \times$ in Fig.~\ref{Fig256E9dBm}. Different initial RIS locations may result in different solutions due to the
inherent non-convexity. However, thanks to the gradient ascent, Algorithm~\ref{Algorithm1}
potentially obtains a good sub-optimal solution. In this paper, we observe that
optimizing the RIS location can significantly improve the coverage probability.
The obtained RIS placement solution can also improve the ergodic rate. A
heuristic method to obtain a good initialization should be left for future work.

\begin{figure}[t]
\begin{minipage}{0.48\textwidth}
		\centering
		\includegraphics[trim=0.2cm 0cm 0.5cm 0.6cm, clip=true, width=3.3in]{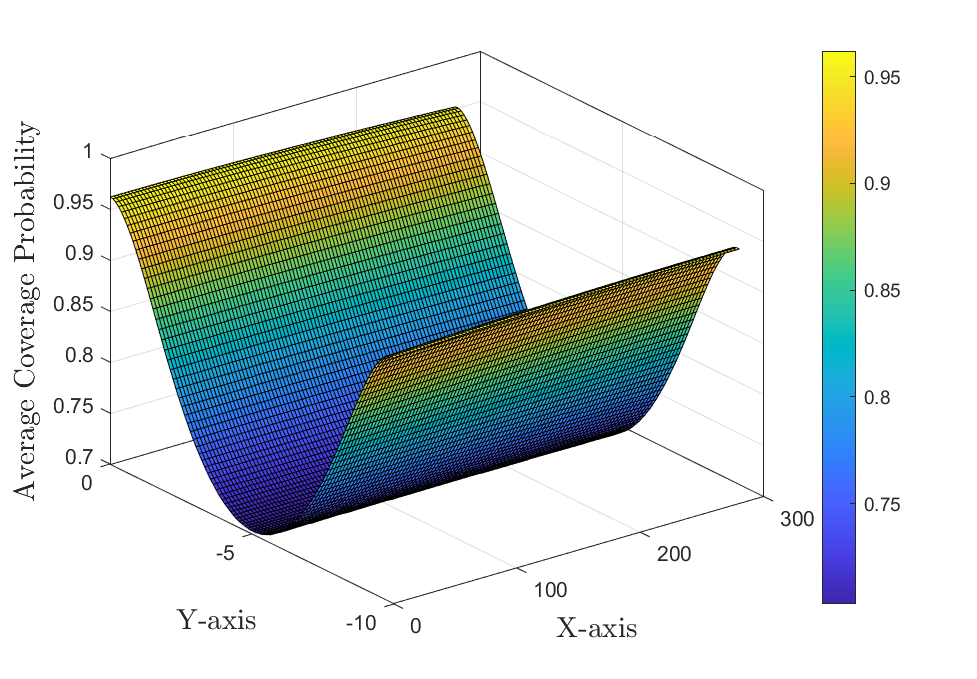} \vspace*{-0.25cm}
        \\ (a)
		\vspace*{-0.0cm}
\end{minipage}
	\vspace*{-0.0cm}
	\hfill
\begin{minipage}{0.48\textwidth}
		\centering
		\includegraphics[trim=0.2cm 0cm 0.9cm 0.6cm, clip=true, width=3.3in]{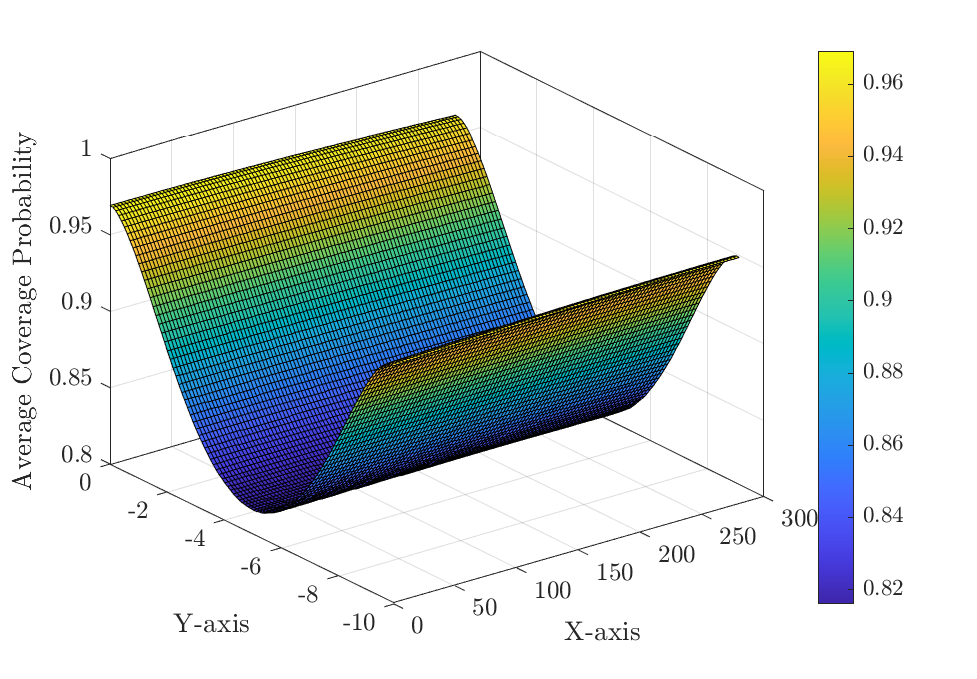} \vspace*{-0.25cm} \\ (b)
		\vspace*{-0.0cm}
  \end{minipage}
  \caption{Average coverage probability with different locations of the RIS where $0$~[m] $\leq x_r \leq$ $280$~[m], $-10$~[m] $\leq y_r \leq$ $0$~[m], and $z_r = 15$~[m]. The destination location: $100$~[m] $\leq x_{\mathrm{d}} \leq $ $180$~[m],  $50$~[m] $\leq y_{\mathrm{d}} \leq $ $100$~[m], $z_{\mathrm{d}} = 15$~[m]: $(a)$ Long-term phase design and $(b)$ Short-term phase shift design.} \label{FigDiffLocations}
\end{figure}

In Fig.~\ref{FigDiffLocations}, we show the average coverage probability as a function of the RIS location. We observe that the RIS should be located near the source or the destination to obtain a good coverage probability. A low coverage probability is, on the other hand, obtained if the RIS is located around the center of the considered area. The observations provide heuristic strategies to locate the RIS and achieve higher coverage probability. However, the optimal location can only be obtained by solving the problem in \eqref{Prob:UAV}.
\section{Conclusion}
This paper has investigated the coverage probability and the
ergodic rate of an RIS-assisted link for different phase shift
designs depending on the level of CSI that
is exploited for optimizing the RIS. Although a long-term phase shift design based on long-term CSI is
suboptimal compared with
the optimal phase shift design based on instantaneous CSI, we have shown that the performance gap decreases as the number of RIS elements increases. Moreover, we have formulated an RIS placement optimization problem that maximizes the coverage probability. A sub-optimal solution with low computational complexity has been proposed by applying the gradient ascent method. Numerical results have shown that, together with the phase shift design, the RIS placement is a potential research direction for improving the system performance. The generalization of this research work includes the analysis and optimization of sources and
destinations equipped with multiple antennas.

\appendix
	\subsection{Useful Lemmas}
	This section presents a useful lemma used for performance analysis. 
	\begin{lemma}\cite[Lemma~5]{van2021reconfigurable} \label{lemma:4momentv1}
		For a random vector $\mathbf{x} \in \mathbb{C}^{M}$ distributed as $\mathcal{CN}(\mathbf{0}, \mathbf{R})$ with $\mathbf{R} \in \mathbb{C}^{M \times M}$ and two given deterministic matrices $\mathbf{U}, \mathbf{V} \in \mathbb{C}^{M \times M}$, it holds that
		\begin{equation} \label{eq:xUxxVxx}
			\mathbb{E} \{\mathbf{x}^H \mathbf{U} \mathbf{x} \mathbf{x}^H \mathbf{V} \mathbf{x}  \} = \mathrm{tr}(\mathbf{R} \mathbf{U} \mathbf{R} \mathbf{V}) + \mathrm{tr}(\mathbf{R}\mathbf{U}) \mathrm{tr}(\mathbf{R} \mathbf{V} ). 
		\end{equation}
		If $\mathbf{V} = \mathbf{U}^H$, then \eqref{eq:xUxxVxx} becomes $\mathbb{E} \{|\mathbf{x}^H \mathbf{U} \mathbf{x}|^2  \} = \mathrm{tr}(\mathbf{R} \mathbf{U} \mathbf{R} \mathbf{U}^H) + |\mathrm{tr}(\mathbf{R}\mathbf{U})|^2$.
	\end{lemma}
	\subsection{Proof of Theorem~\ref{Theorem:RISChannel}} \label{Appendix:RISChannel}
	By exploiting the channel formulation in \eqref{eq:ChannelMod}, the second expectation in \eqref{eq:SNRranMean} is computed as
	\begin{equation} \label{eq:Expcascaded2}
		\begin{split}
			&\mathbb{E} \left\{ \mathbf{h}_{\mathrm{sr}}^H \pmb{\Phi} \mathbf{h}_{\mathrm{rd}}|^2 \right\} = \mathbb{E} \left\{ (\bar{\mathbf{h}}_{\mathrm{sr}}^H +  \mathbf{g}_{\mathrm{sr}}^H ) \pmb{\Phi} (\bar{\mathbf{h}}_{\mathrm{rd}} + \mathbf{g}_{\mathrm{rd}})|^2 \right\}\\
			&\stackrel{(a)}{=} |\bar{\alpha}|^2 + \mathbb{E} \big\{ |\bar{\mathbf{h}}_{\mathrm{sr}}^H \pmb{\Phi} \mathbf{g}_{\mathrm{rd}}|^2 \big\} +  \mathbb{E} \big\{ |\mathbf{g}_{\mathrm{sr}}^H \pmb{\Phi} \bar{\mathbf{h}}_{\mathrm{rd}}|^2 \big\}  +  \mathbb{E} \big\{ |\mathbf{g}_{\mathrm{sr}}^H \pmb{\Phi} \mathbf{g}_{\mathrm{rd}}|^2 \big\} \\
			&\stackrel{(b)}{=} |\bar{\alpha}|^2 + \frac{\beta_{\mathrm{rd}}}{\kappa_{\mathrm{rd}}+1} \| \bar{\mathbf{h}}_{\mathrm{sr}}\|^2 + \frac{\beta_{\mathrm{sr}}}{K_{\mathrm{sr}}+1} \| \bar{\mathbf{h}}_{\mathrm{rd}}\|^2  + \frac{M\beta_{\mathrm{sr}} \beta_{\mathrm{rd}}}{(\kappa_{\mathrm{sr}}+1)(\kappa_{\mathrm{rd}}+1)}\\
           &= |\bar{\alpha}|^2 + M\mu\beta_{\mathrm{sr}}\beta_{\mathrm{rd}}\tilde{\kappa}.
		\end{split}
	\end{equation}
	where $(a)$ is obtained by the independence of the cascaded channels and $(b)$ is due to the distributions of the NLoS channels.
	By utilizing the definition of cascaded channels, the fourth moment in \eqref{eq:4MomentCascade} is recast as 
	\begin{equation} \label{eq:Cascade4v1}
		\begin{split}
			&\mathbb{E} \left\{ |\mathbf{h}_{\mathrm{sr}}^H \pmb{\Phi} \mathbf{h}_{\mathrm{rd}}|^4 \right\} = \mathbb{E} \left\{ (\bar{\mathbf{h}}_{\mathrm{sr}}^H +  \mathbf{g}_{\mathrm{sr}}^H ) \pmb{\Phi} (\bar{\mathbf{h}}_{\mathrm{rd}} + \mathbf{g}_{\mathrm{rd}})|^4 \right\} \\
			&= \mathbb{E}\left\{ \left| \underbrace{\bar{\mathbf{h}}_{\mathrm{sr}}^H \pmb{\Phi} \bar{\mathbf{h}}_{\mathrm{rd}}}_{\triangleq \bar{\alpha}} +  \underbrace{\bar{\mathbf{h}}_{\mathrm{sr}}^H \pmb{\Phi} \mathbf{g}_{\mathrm{rd}}}_{\triangleq \alpha_2 } +   \underbrace{\mathbf{g}_{\mathrm{sr}}^H \pmb{\Phi} \bar{\mathbf{h}}_{\mathrm{rd}}}_{\triangleq \alpha_3} +   \underbrace{\mathbf{g}_{\mathrm{sr}}^H \pmb{\Phi} \mathbf{g}_{\mathrm{rd}}}_{\triangleq \alpha_4 } \right|^4 \right\}.
		\end{split}
	\end{equation}
	Denoting $A = \alpha_2 + \alpha_3 + \alpha_4$, \eqref{eq:Cascade4v1} is equivalent to
	\begin{equation} \label{eq:Cascadev2}
		\begin{split}
			&\mathbb{E} \left\{ |\mathbf{h}_{\mathrm{sr}}^H \pmb{\Phi} \mathbf{h}_{\mathrm{rd}}|^4 \right\} =  \mathbb{E}\{ |(\bar{\alpha} + A)(\bar{\alpha}^\ast + A^\ast)|^2\} \\
   & = \mathbb{E}\{ | |\bar{\alpha}|^2 + \bar{\alpha}^\ast A + A^\ast \bar{\alpha} + |A|^2 |^2 \}\\
			& =  |\bar{\alpha}|^4 + 4 |\bar{\alpha}|^2 \mathbb{E} \{ |A|^2 \} + 2 \bar{\alpha} \mathbb{E} \{ A^\ast |A|^2 \} + 2 \bar{\alpha}^\ast \mathbb{E} \{ A |A|^2 \}  \\
   &\quad + \mathbb{E} \{ |A|^4 \}.
		\end{split}
	\end{equation}
	  We note that $|A|^2$ can be rewritten as follows:
	\begin{equation} \label{eq:absA2}
		\begin{split}
		&|A|^2 = (\alpha_2 + \alpha_3 + \alpha_4)(\alpha_2^\ast + \alpha_3^\ast + \alpha_4^\ast) = |\alpha_2|^2 + |\alpha_3|^2 + \\
		& |\alpha_4|^2  + \alpha_2 \alpha_3^\ast + \alpha_2 \alpha_4^\ast + \alpha_2^\ast \alpha_3 + \alpha_3 \alpha_4^\ast  + \alpha_2^\ast \alpha_4 + \alpha_3^\ast \alpha_4 ,
		\end{split}
	\end{equation}
	and therefore $\mathbb{E} \{ |A|^2 \}$ is reformulated as
	\begin{equation} \label{eq:A2}
		\mathbb{E} \{ |A|^2 \} = \mathbb{E} \{ | \alpha_2^2 |\} + \mathbb{E} \{ | \alpha_3^2 |\} + \mathbb{E} \{ | \alpha_4^2 |\},
	\end{equation}
	where the missing expectations are equal to zero. Since $\mathbf{g}_{\mathrm{sr}}$ and $\mathbf{g}_{\mathrm{rd}}$ are circularly symmetric complex Gaussian vectors with zero mean, we obtain 
	\begin{align}
		\mathbb{E} \{|\alpha_2|^2 \}  &= \bar{\mathbf{h}}_{\mathrm{sr}}^H \pmb{\Phi} \mathbb{E}\{\mathbf{g}_{\mathrm{rd}} \mathbf{g}_{\mathrm{rd}}^H \} \pmb{\Phi}^H \bar{\mathbf{h}}_{\mathrm{sr}} = \frac{\beta_{\mathrm{rd}}}{\kappa_{\mathrm{rd}}+1} \| \bar{\mathbf{h}}_{\mathrm{sr}} \|^2 = M\mu_n \kappa_{\mathrm{sr}},  \label{eq:Expa2}\\
		\mathbb{E} \{|\alpha_3|^2 \}  &= \bar{\mathbf{h}}_{\mathrm{rd}}^H \pmb{\Phi} \mathbb{E}\{\mathbf{g}_{\mathrm{sr}} \mathbf{g}_{\mathrm{sr}}^H \} \pmb{\Phi}^H \bar{\mathbf{h}}_{\mathrm{rd}} = \frac{\beta_{\mathrm{sr}}}{\kappa_{\mathrm{sr}}+1} \|\bar{\mathbf{h}}_{\mathrm{rd}}\|^2 = M\mu \kappa_{\mathrm{rd}}, \label{eq:Exp3}
	\end{align}
	by utilizing the identity $\mathrm{tr}(\mathbf{X}\mathbf{Y}) = \mathrm{tr}(\mathbf{Y}\mathbf{X})$, $\mathbb{E} \{|\alpha_4|^2 \}$ can be computed as follows
	\begin{equation} \label{eq:Exp4}
		\mathbb{E} \{ |\alpha_4|^2 \} \stackrel{(a)}{=} \mathbb{E} \left\{ \mathbf{g}_{\mathrm{s_1r}}^H \pmb{\Phi} \mathbb{E} \{\mathbf{g}_{\mathrm{rd}} \mathbf{g}_{\mathrm{rd}}^H \} \pmb{\Phi}^H \mathbf{g}_{\mathrm{s_1r}} \right\} = M \mu, 
	\end{equation}
	where $(a)$ is obtained since the cascaded channels are independent. By plugging \eqref{eq:Expa2}--\eqref{eq:Exp4} into \eqref{eq:A2} and doing some algebra, we obtain the closed form expression of $\mathbb{E}\{ |A|^2 \}$ as
	\begin{equation}
		\mathbb{E} \{ |A|^2 \} =  M\mu \tilde{\kappa}.
	\end{equation}
	The expectation $\mathbb{E}\{ A^\ast |A|^2 \}$ in \eqref{eq:Cascadev2} can be simplified by utilizing the definition of $A$ and \eqref{eq:absA2} as
	\begin{equation}
		\begin{split}
		 &	\mathbb{E}\{ A^\ast |A|^2 \}  \stackrel{(a)}{=} 2 \mathbb{E}\{ \alpha_2^\ast  \alpha_3^\ast  \alpha_4 \} = 2 \mathbb{E} \{ \mathbf{g}_{\mathrm{rd}}^H \pmb{\Phi}^H \bar{\mathbf{h}}_{\mathrm{sr}} \bar{\mathbf{h}}_{\mathrm{rd}}^H \pmb{\Phi}^H \mathbf{g}_{\mathrm{sr}} \mathbf{g}_{\mathrm{sr}}^H \pmb{\Phi} \mathbf{g}_{\mathrm{rd}} \}  \\
		 & = \frac{2\beta_{\mathrm{sr}} \beta_{\mathrm{rd}} \mathrm{tr}( \pmb{\Phi}^H \bar{\mathbf{h}}_{\mathrm{sr}} \bar{\mathbf{h}}_{\mathrm{rd}}^H )}{(\kappa_{\mathrm{sr}}+1)(\kappa_{\mathrm{rd}}+1)} = 2 \mu \bar{\alpha}^\ast,
		\end{split}
	\end{equation}
	where $(a)$ is obtained from the products of uncorrelated random variables. Similarly, the expectation $\mathbb{E}\{ A |A|^2 \}$ in \eqref{eq:Cascadev2} is simplified as
	\begin{equation}
		\begin{split}
			&\mathbb{E}\{ A |A|^2 \} = 2 \mathbb{E}\{ \alpha_2  \alpha_3 \alpha_4^\ast \} = 2 \mathbb{E} \{ \bar{\mathbf{h}}_{\mathrm{sr}}^H \pmb{\Phi} \mathbf{g}_{\mathrm{rd}} \mathbf{g}_{\mathrm{sr}}^H \pmb{\Phi}  \bar{\mathbf{h}}_{\mathrm{rd}} \mathbf{g}_{\mathrm{rd}}^H \pmb{\Phi}^H \mathbf{g}_{\mathrm{sr}} \} \\
			& \stackrel{(a)}{=}   \mathbb{E} \{  \mathbf{g}_{\mathrm{sr}}^H \pmb{\Phi}   \bar{\mathbf{h}}_{\mathrm{rd}} \bar{\mathbf{h}}_{\mathrm{sr}}^H \pmb{\Phi} \mathbf{g}_{\mathrm{rd}} \mathbf{g}_{\mathrm{rd}}^H \pmb{\Phi}^H \mathbf{g}_{\mathrm{sr}} \}  =  \frac{2\beta_{\mathrm{sr}} \beta_{\mathrm{rd}} \mathrm{tr}( \pmb{\Phi} \bar{\mathbf{h}}_{\mathrm{rd}} \bar{\mathbf{h}}_{\mathrm{sr}}^H )}{(\kappa_{\mathrm{sr}}+1)(\kappa_{\mathrm{rd}}+1)} = 2 \mu \bar{\alpha},
		\end{split}
	\end{equation}
	where $(a)$ is obtained by using the commutative property two random variables $X$ and $Y$, i.e., $\mathbb{E}\{XY\} = \mathbb{E}\{YX\}$. The expectation $\mathbb{E}\{ |A|^4\}$ in \eqref{eq:Cascadev2} is simplified as
	\begin{equation} \label{eq:ExA4}
	\begin{split}
		\mathbb{E}\{ |A|^4\} =& \mathbb{E} \{ |\alpha_2|^4 \} + \mathbb{E} \{ |\alpha_3|^4 \} + \mathbb{E} \{ |\alpha_4|^4 \} + 4 \mathbb{E} \{ |\alpha_2|^2 |\alpha_3|^2 \} \\
		&+ 4 \mathbb{E} \{ |\alpha_2|^2 |\alpha_4|^2 \} +  4 \mathbb{E} \{ |\alpha_3|^2 |\alpha_4|^2 \},
	\end{split}
	\end{equation}
	where the missing expectations are equal to zero. The first expectation $\mathbb{E} \{ |\alpha_2|^4 \}$ in \eqref{eq:ExA4} is computed in a closed-form expression as
	\begin{equation}
		\begin{split}
			 \mathbb{E} \{ |\alpha_2|^4 \} & = \mathbb{E} \{ |\bar{\mathbf{h}}_{\mathrm{sr}}^H \pmb{\Phi} \mathbf{g}_{\mathrm{rd}}|^4 \} = \mathbb{E} \{ | \mathbf{g}_{\mathrm{rd}}^H \pmb{\Phi}^H \bar{\mathbf{h}}_{\mathrm{sr}} \bar{\mathbf{h}}_{\mathrm{sr}}^H \pmb{\Phi} \mathbf{g}_{\mathrm{rd}} |^2 \} \\
   & = 2M^2\mu^2 K_{\mathrm{sr}}^2,  
		\end{split}
	\end{equation}
	where $(a)$ follows from  Lemma~\ref{lemma:4momentv1}. Similarly, the  expectation $\mathbb{E} \{ |\alpha_3|^4 \}$ in \eqref{eq:ExA4} is computed in a closed-form expression as
	\begin{equation} \label{eq:Expalpha2}
		\mathbb{E} \{ |\alpha_3|^4 \} = \mathbb{E} \{ |\mathbf{g}_{\mathrm{s_1r}}^H \pmb{\Phi} \bar{\mathbf{h}}_{\mathrm{rd}} \bar{\mathbf{h}}_{\mathrm{rd}}^H \pmb{\Phi}^H \mathbf{g}_{\mathrm{s_1r}} |^2 \} = 2M^2\mu^2  K_{\mathrm{rd}}^2.
	\end{equation}
	The third expectation $\mathbb{E} \{ |\alpha_3|^4 \}$ in \eqref{eq:ExA4} is simplified as
	\begin{equation} \label{eq:alpha4}
		\begin{split}
		&\mathbb{E} \{ |\alpha_4|^4  \} = \mathbb{E}\left\{ |\mathbf{g}_{\mathrm{sr}}^H \pmb{\Phi} \mathbf{g}_{\mathrm{rd}}|^4 \right\} = \mathbb{E} \left\{ \| \pmb{\Phi} \mathbf{g}_{\mathrm{rd}} \|^4 \left|\frac{\mathbf{g}_{\mathrm{sr}}^H \pmb{\Phi} \mathbf{g}_{\mathrm{rd}}}{\| \pmb{\Phi} \mathbf{g}_{\mathrm{rd}} \|}  \frac{\mathbf{g}_{\mathrm{rd}}^H \pmb{\Phi}^H \mathbf{g}_{\mathrm{sr}}}{\| \pmb{\Phi} \mathbf{g}_{\mathrm{rd}} \|} \right|^2 \right\}.
		\end{split}
	\end{equation}
	By introducing the random variable $z = \mathbf{g}_{\mathrm{sr}}^H \pmb{\Phi} \mathbf{g}_{\mathrm{rd}}/\| \pmb{\Phi} \mathbf{g}_{\mathrm{rd}} \| $ with $z \sim \mathcal{CN}(0,\beta_{\mathrm{sr}}/(\kappa_{\mathrm{sr}}+1) )$, we see that \eqref{eq:alpha4} is equivalent to 
	\begin{equation} \label{eq:4moment}
	 \begin{split}
		& \mathbb{E} \{ |\alpha_4|^4  \} = \mathbb{E} \left\{ \| \pmb{\Phi} \mathbf{g}_{\mathrm{rd}} \|^4 |z|^4 \right\} \stackrel{(a)}{=} \mathbb{E} \left\{ \| \pmb{\Phi} \mathbf{g}_{\mathrm{rd}} \|^4 \right\} \mathbb{E} \{  |z|^4 \} \\
    & = 2 (M^2 +M) \mu^2, 
	\end{split}
	\end{equation}
	where $(a)$ is obtained by the independence of $\pmb{\Phi} \mathbf{g}_{\mathrm{rd}}$ and $z$, and the fourth moments of the circularly symmetric Gaussian variables with zero means in \eqref{eq:4moment} is computed by utilizing Lemma~\ref{lemma:4momentv1}. Buy applying the results in \eqref{eq:Expa2}, the forth expectation in \eqref{eq:ExA4} is simplified  as follows
	\begin{equation}
		\mathbb{E} \{ |\alpha_2|^2 |\alpha_3|^2 \} \stackrel{(a)}{=}   \mathbb{E} \{ |\alpha_2|^2 \}  \mathbb{E} \{ |\alpha_3|^2 \} = M^2 \mu^2 \kappa_{\mathrm{sr}} \kappa_{\mathrm{rd}},
	\end{equation}
	where $(a)$ follows since the cascaded channels are independent. The fifth expectation in \eqref{eq:ExA4} can be simplified as follows
	\begin{equation}
		\begin{split}
			&\mathbb{E} \{ |\alpha_2|^2 |\alpha_4|^2 \} =  \mathbb{E} \{ \mathbf{g}_{\mathrm{rd}}^H \pmb{\Phi}^H \bar{\mathbf{h}}_{\mathrm{sr}} \bar{\mathbf{h}}_{\mathrm{sr}}^H \pmb{\Phi} \mathbf{g}_{\mathrm{rd}} \mathbf{g}_{\mathrm{rd}}^H \pmb{\Phi}^H \mathbf{g}_{\mathrm{sr}} \mathbf{g}_{\mathrm{sr}}^H \pmb{\Phi} \mathbf{g}_{\mathrm{rd}}  \} \\
			& \stackrel{(a)}{=} \frac{\beta_{\mathrm{rd}}^2}{(\kappa_{\mathrm{rd}} +1)^2}\mathrm{tr}(\pmb{\Phi}^H \bar{\mathbf{h}}_{\mathrm{sr}} \bar{\mathbf{h}}_{\mathrm{sr}}^H \pmb{\Phi} \pmb{\Phi}^H \mathbb{E}\{\mathbf{g}_{\mathrm{sr}} \mathbf{g}_{\mathrm{sr}}^H \} \pmb{\Phi})  + \\
   &\qquad \frac{\beta_{\mathrm{rd}}^2}{(\kappa_{\mathrm{rd}} +1)^2} \mathrm{tr}(\pmb{\Phi}^H \bar{\mathbf{h}}_{\mathrm{sr}} \bar{\mathbf{h}}_{\mathrm{sr}}^H \pmb{\Phi}) \mathrm{tr}( \pmb{\Phi}^H \mathbb{E}\{\mathbf{g}_{\mathrm{sr}} \mathbf{g}_{\mathrm{sr}}^H \} \pmb{\Phi}) \\
			& \stackrel{(b)}{=} (M^2 +M)\mu^2 \kappa_{\mathrm{sr}},
		\end{split}
	\end{equation}
	where $(a)$ is obtained by using Lemma~\ref{lemma:4momentv1} and $(b)$ follows from the identity $\pmb{\Phi} \pmb{\Phi}^H = \mathbf{I}_M$. The  last expectation in \eqref{eq:ExA4} is derived  as follows
	\begin{align}
		\mathbb{E} \{ |\alpha_3|^2 |\alpha_4|^2 \} &=  \mathbb{E} \{ |\alpha_3|^2 \} \mathbb{E} \{ |\alpha_4|^2 \} =  (M^2 +M)\mu^2 K_{\mathrm{rd}}
	\end{align}
	Plugging \eqref{eq:Expa2}, \eqref{eq:Exp3}, and \eqref{eq:Exp4} into \eqref{eq:4moment} and with some algebraic manipulations, we prove the theorem.
	\subsection{Proof of Lemma~\ref{Corollary:LTOpt}} \label{Appendix:LTOpt} 
	We first consider the long-term phase shift design. By exploiting the definition of the SNR  in \eqref{eq:SNRLTv}, we compute the mean of $\gamma^{\mathrm{lt}}$ as follows:
	\begin{equation} \label{eq:SNRranMean}
		\mathbb{E} \{ \gamma^{\mathrm{lt}} \} = \nu \mathbb{E} \left\{ |h_{\mathrm{sd}}|^2 \right\} + \nu \mathbb{E} \left\{ |\mathbf{h}_{\mathrm{sr}}^H \pmb{\Phi}^{\mathrm{lt}}  \mathbf{h}_{\mathrm{rd}}|^2 \right\},
	\end{equation}
	thanks to the independence of the direct link and the cascaded channels. Since the direct link follows a Rayleigh distribution, the first expectation in \eqref{eq:SNRranMean} is computed in closed-form expression as
	\begin{equation} \label{eq:Ehsdn}
		\mathbb{E} \{ |h_{\mathrm{sd}}|^2 \} = \beta_{\mathrm{sd}},
	\end{equation}
	while the second expectation is computed as in \eqref{eq:2MomentCascade}. Plugging \eqref{eq:Ehsdn} and \eqref{eq:2MomentCascade} into \eqref{eq:SNRranMean}, we obtain the closed-form expression of the mean of the SNR as
	\begin{equation} \label{eq:Expgamma}
		\mathbb{E} \{ \gamma^{\mathrm{lt}} \} = \nu (\beta_{\mathrm{sd}} + \delta ).
	\end{equation}
	From \eqref{eq:Expgamma}, the long-term phase shift design to maximize the average received SNR  is obtained from the following optimization problem
	\begin{equation} 
	\begin{aligned}
		&\underset{\pmb{\Phi}^{\mathsf{lt}}}{\mathrm{maximize}} &\quad &  |\bar{\alpha} |^2  \\
		&\mbox{subject to} &&  - \pi \leq \theta_{m}^{\mathsf{lt}} \leq \pi, \forall m, 
	\end{aligned}
   \end{equation}
 where $|\bar{\alpha} |^2$ is given in Theorem~\ref{Theorem:RISChannel}. Decomposing the LoS components into the magnitude and phase as 
	\begin{align}
		&[\bar{\mathbf{h}}_{\mathrm{sr}}^\ast]_m =  |[\bar{\mathbf{h}}_{\mathrm{sr}}^\ast]_m| e^{j\arg([\bar{\mathbf{h}}_{\mathrm{sr}}^\ast]_m)},  \forall m, \\
  & [\bar{\mathbf{h}}_{\mathrm{rd}}]_m =  |[\bar{\mathbf{h}}_{\mathrm{rd}}]_m| e^{j\arg([\bar{\mathbf{h}}_{\mathrm{rd}}]_m)}, \forall m,
	\end{align}
	  $|\bar{\alpha}|^2$ is reformulated  as
	\begin{equation} \label{eq:alphabar}
		|\bar{\alpha}|^2 = \left| \sum\nolimits_{m=1}^M  |[\bar{\mathbf{h}}_{\mathrm{rd}}]_m| |[\bar{\mathbf{h}}_{\mathrm{sr}}^\ast]_m| e^{j (\arg([\bar{\mathbf{h}}_{\mathrm{sr}}^\ast]_m) + \arg([\bar{\mathbf{h}}_{\mathrm{rd}}]_m) + \theta_{m}^{\mathsf{lt}} )}  \right|^2.
	\end{equation}
	Next, let us denote $\tilde{\theta}_{m}^{\mathsf{lt}} = \arg([\bar{\mathbf{h}}_{\mathrm{sr}}^\ast]_m) + \arg([\bar{\mathbf{h}}_{\mathrm{rd}}]_m) + \theta_{m}^{\mathsf{lt}}, \forall m, $ and introduce the two vectors $\mathbf{a} \in \mathbb{C}^M$ and $\mathbf{b} \in \mathbb{C}^M$ as follows
	\begin{align}
		\mathbf{a} &= \left[\sqrt{|[\bar{\mathbf{h}}_{\mathrm{rd}}]_1| |[\bar{\mathbf{h}}_{\mathrm{sr}}^\ast]_1|}, \ldots, \sqrt{|[\bar{\mathbf{h}}_{\mathrm{rd}}]_m| |[\bar{\mathbf{h}}_{\mathrm{sr}}^\ast]_M|} \right]^T, \\
		\mathbf{b} &= \left[\sqrt{|[\bar{\mathbf{h}}_{\mathrm{rd}}]_1| |[\bar{\mathbf{h}}_{\mathrm{sr}}^\ast]_1|}  e^{j \tilde{\theta}_{1}^{\mathsf{lt}}}, \ldots, \sqrt{|[\bar{\mathbf{h}}_{\mathrm{rd}}]_m| |[\bar{\mathbf{h}}_{\mathrm{sr}}^\ast]_M|}  e^{j \tilde{\theta}_{M}^{\mathsf{lt}}} \right]^T.
	\end{align}
	Then, \eqref{eq:alphabar} can be rewritten in the equivalent form as follows
	\begin{equation}
		|\bar{\alpha}|^2 = | \mathbf{a}^H \mathbf{b} |^2 \stackrel{(a)}{\leq} \| \mathbf{a} \|^2 \| \mathbf{b} \|^2 = \left(   \sum\nolimits_{m=1}^M  |[\bar{\mathbf{h}}_{\mathrm{rd}}]_m| |[\bar{\mathbf{h}}_{\mathrm{sr}}^\ast]_m|  \right)^2,
	\end{equation}
	where $(a)$ is obtained by utilizing Cauchy-Schwarz's inequality, where the equality holds if the two vectors $\mathbf{a}$ and $\mathbf{b}$ are parallel to each other. For all the phase shifts $\tilde{\theta}_{m}^{\mathsf{lt}}$, it holds that
	\begin{equation}
		\tilde{\theta}_{m}^{\mathsf{lt}} = 0 \Leftrightarrow  \theta_{m}^{\mathsf{lt}}  + \arg([\bar{\mathbf{h}}_{\mathrm{sr}}^\ast]_m) + \arg([\bar{\mathbf{h}}_{\mathrm{rd}}]_m) = 0,
	\end{equation}
	and therefore, the optimal long-term phase shift design is given in \eqref{eq:Optphasse}.

In a similar manner, for the short-term design, let us decompose the complex channel coefficients into their modulus and phase as follows
 \begin{multline} \label{eq:ShortTerm}
 |h_{\mathrm{sd}} + \mathbf{h}_{\mathrm{sr}}^H \pmb{\Phi}^{\mathsf{st}} \mathbf{h}_{\mathrm{rd}}|^2 =  \left| |h_{\mathrm{sd}}| e^{j \arg(h_{\mathrm{sd}}) } + \right. \\ \left. \sum\nolimits_{m=1}^M [\mathbf{h}_{\mathrm{sr}}^\ast]_m  [\mathbf{h}_{\mathrm{rd}}]_m e^{j (  \arg(\theta_{m}^{\mathsf{lt}})+ \arg([\mathbf{h}_{\mathrm{sr}}^\ast]_m ) + \arg([\mathbf{h}_{\mathrm{rd}}]_m )} \right|^2.
 \end{multline}
 Let us introduce the two vectors $\tilde{\mathbf{a}} \in \mathbb{C}^{M+1}$ and $\tilde{\mathbf{b}} \in \mathbb{C}^{M+1}$ as follows
 \begin{align}
 &\tilde{\mathbf{a}} = \left[ \sqrt{|h_{\mathrm{sd}}|},  \sqrt{[\mathbf{h}_{\mathrm{sr}}^\ast]_1  [\mathbf{h}_{\mathrm{rd}}]_1}, \ldots, \sqrt{[\mathbf{h}_{\mathrm{sr}}^\ast]_M  [\mathbf{h}_{\mathrm{rd}}]_M}  \right]^T ,\\
 & \tilde{\mathbf{b}} =  \left[ \sqrt{|h_{\mathrm{sd}}|} e^{j \arg(h_{\mathrm{sd}}) },  \sqrt{[\mathbf{h}_{\mathrm{sr}}^\ast]_1  [\mathbf{h}_{\mathrm{rd}}]_1} e^{j \arg(\hat{\theta}_{1}^{\mathsf{lt}})}, \ldots, \right. \notag \\ 
 & \qquad \left. \sqrt{[\mathbf{h}_{\mathrm{sr}}^\ast]_M  [\mathbf{h}_{\mathrm{rd}}]_M} e^{j \arg(\hat{\theta}_{M}^{\mathsf{lt}})}  \right]^T,
 \end{align}
where $\arg(\hat{\theta}_{m}^{\mathsf{lt}}) = \arg(\theta_{m}^{\mathsf{lt}})+ \arg([\mathbf{h}_{\mathrm{sr}}^\ast]_m ) + \arg([\mathbf{h}_{\mathrm{rd}}]_m), \forall m$. Accordingly, \eqref{eq:ShortTerm} can be reformulated in the equivalent form as follows
\begin{equation}
\begin{split}
&|h_{\mathrm{sd}} + \mathbf{h}_{\mathrm{sr}}^H \pmb{\Phi}^{\mathsf{st}} \mathbf{h}_{\mathrm{rd}}|^2 = |\tilde{\mathbf{a}}^H \tilde{\mathbf{b}}|^2 \stackrel{(a)}{\leq} \| \tilde{\mathbf{a}} \|^2 \| \tilde{\mathbf{b}} \|^2 \\
& = \left( |h_{\mathrm{sd}} | +  \sum\nolimits_{m=1}^M  |[\bar{\mathbf{h}}_{\mathrm{rd}}]_m| |[\bar{\mathbf{h}}_{\mathrm{sr}}^\ast]_m|  \right)^2 ,
\end{split}
\end{equation}
where $(a)$ is attained by using Cauchy-Schwarz's inequality, where the equality holds if $\tilde{\mathbf{a}}$ and $\tilde{\mathbf{b}}$ are parallel to each other, i.e.,
\begin{equation}
 \arg(h_{\mathrm{sd}})  =  \arg(\theta_{m}^{\mathsf{lt}})+ \arg([\mathbf{h}_{\mathrm{sr}}^\ast]_m ) + \arg([\mathbf{h}_{\mathrm{rd}}]_m).
\end{equation}
Therefore, the optimal short-term phase shift design is obtained as in \eqref{eq:OptSt} and we complete the proof.
	\subsection{Proof of Theorem~\ref{Theorem:CovProbRan}}\label{Appendix:CovProbRan}
	The average received SNR  for the long-term phase shift design is given in \eqref{eq:Expgamma}. Denoting $a = \sqrt{\nu} h_{\mathrm{sd}}$ and $b = \sqrt{\nu}\mathbf{h}_{\mathrm{sr}}^H \pmb{\Phi}^{\mathsf{lt}} \mathbf{h}_{\mathrm{rd}}$, we recast the second moment of the SNR as follows
	\begin{equation} \label{eq:GammaRan2}
		\begin{split}
			& \mathbb{E} \{ |\gamma^{\mathrm{lt}}|^2 \} = \mathbb{E} \big\{ \big| |a|^2  + a^\ast b + a b^\ast + |b|^2 \big|^2 \big\} \\
   &= \mathbb{E} \{ |a|^4 \} +   \mathbb{E} \{ |b|^4 \}   + 4 \mathbb{E} \{ |a|^2 |b|^2 \}.
		\end{split}
	\end{equation}
	By exploiting Lemma~\ref{lemma:4momentv1}, we tackle the first expectation  in \eqref{eq:GammaRan2} as follows:
	\begin{equation}
		\mathbb{E} \{ |a|^4 \} = \nu^2 \mathbb{E} \{ |h_{\mathrm{sd}}|^4 \} =  2 \nu^2 \beta_{\mathrm{sd}}^2.
	\end{equation}
	From the result in \eqref{eq:4MomentCascade}, the second expectation in \eqref{eq:GammaRan2} is computed in closed-form expression as follows:
	\begin{equation}
		\mathbb{E} \{ |b|^4 \} =  \nu^2 (\delta^2  + 2M |\bar{\alpha}|^2 \mu  \widetilde{K} + M^2 \mu^2 \tilde{\kappa}^2 +  2M \mu^2 \hat{\kappa} +8 |\bar{\alpha}|^2 \mu). 
	\end{equation}
	Next, the fourth expectation in \eqref{eq:GammaRan2} is computed in  closed-form expression as follows
	\begin{equation}
		\begin{split}
			&\mathbb{E}\{ |a|^2 |b|^2 \} = \nu^2 \mathbb{E}\{ |h_{\mathrm{sd}}|^2 \} \mathbb{E} \left\{ \mathbf{h}_{\mathrm{sr}}^H \pmb{\Phi}^{\mathsf{lt}}  \mathbf{h}_{\mathrm{rd}}|^2 \right\} = \nu^2 \beta_{\mathrm{sd}} \delta,
		\end{split}
	\end{equation}
	where $(a)$ is obtained by the independence between the direct and indirect links; $(b)$ follows from  \eqref{eq:Expcascaded2}. Combing \eqref{eq:Expgamma} and \eqref{eq:GammaRan2} together with the identity $\mathsf{Var}\{X \} = \mathbb{E}\{ |X|^2\} - |\mathbb{E}\{ X\}|^2$, we obtain
	\begin{equation} \label{eq:Valgamma}
		\mathsf{Var}\{\gamma^{\mathrm{lt}} \} =   \nu^2 \beta_{\mathrm{sd}}^2 + \nu^2 \tilde{a} + 2 \nu^2 \beta_{\mathrm{sd}} \delta.
	\end{equation}
From the  mean in \eqref{eq:SNRranMean} and the variance in \eqref{eq:Valgamma}, we match the received SNR, $\gamma^{\mathrm{lt}}$, to a Gamma distribution and obtain the result as shown in the theorem. 
\subsection{Proof of Theorem~\ref{Theorem:CovProbOpt}}\label{Appendix:CovProbOpt}
In this section, we compute the coverage probability for the short-term phase shift design using the double-matching method. We start with some properties of the Rice distribution. Particularly, the mean and variance of each channel gain $\left| \left[ \mathbf{h}_{\mathrm{sr}} \right]_m \right|$, which follows by Rice distribution, are given by
	\begin{align} 
		&\mathbb{E} \left\{ {\left| {{{\left[ {\mathbf{h}_{\mathrm{sr}}} \right]}_m}} \right|} \right\} = \frac{1}{2}\sqrt {\frac{\pi \beta _{\mathrm{sr}}}{{{\kappa_{\mathrm{sr}}} + 1}}} t_{\mathrm{sr}}, \\
		&\mathsf{Var}\left\{ \left| [ \mathbf{h}_{\mathrm{sr}} ]_m \right| \right\} = 
		\frac{\beta _{\mathrm{sr}}}{\kappa_{\mathrm{sr}} + 1}\left( 1 + \kappa_{\mathrm{sr}} - \frac{\pi }{4} t_{\mathrm{sr}}^2 \right).
	\end{align}
	With the presence of the RIS, the mean and variance of $B_{m} = | [ \mathbf{h}_{\mathrm{sr}}]_m | [ \mathbf{h}_{\mathrm{rd}} ]_m|$ is then computed as follows
	\begin{align}
		\mathbb{E}\{ B_{m}\} =& \frac{\pi}{4} \sqrt{\mu} t_{\mathrm{sr}} t_{\mathrm{rd}} ,
		 \\
		\mathsf{Var} \left\{ {{B_{m}}} \right\} =& \mu
		\left( \left( {1 + {\kappa_{\mathrm{sr}}}} \right)\left( {1 + {\kappa_{\mathrm{rd}}}} \right) 
		- \frac{{{\pi ^2}}}{{16}}{ t_{\mathrm{sr}}^2 t_{\mathrm{rd}}^2 } \right).
	\end{align}
	Also, the mean and variance of the  random variable $C = | h_{\mathrm{sd}} | + \sum\nolimits_{m = 1}^M \left| {{{[ {{\mathbf{h}_{\mathrm{sr}}}} ]}_m}} \right|\left| {{{\left[ {{\mathbf{h}_{\mathrm{rd}}}} \right]}_m}} \right|$ are given as follows
	\begin{align}
		& \mathbb{E} \{ C \} = \frac{1}{2}\sqrt {\pi {\beta_{\mathrm{sd}}}}  +    \sum\nolimits_{m = 1}^M  \mathbb{E} \left\{ {{B_{m}}} \right\}, \\
		&\mathsf{Var} \{ C \} = \frac{{4 - \pi }}{4}{\beta_{\mathrm{sd}}} +      \sum\nolimits_{m = 1}^M \mathsf{Var} \left\{ {{B_{m}}} \right\}
	\end{align}
	Given the mean and variance of $C$, we  
	apply the first-moment matching. More precisely, we match the random variable $C$ to a Gamma distribution as follows
	\begin{align}
		k_{C} =& \frac{{{{\left| {\mathbb{E} \{ C\}} \right|}^2}}}{{\mathsf{Var}\left\{ C \right\}}} \mbox{ and } {w_{C}} = \frac{{\mathsf{Var}\left\{ C \right\}}}{{ \mathbb{E} \left\{ C \right\}}},
	\end{align}
	for which the solutions to $k_{C}$ and $w_{C}$ are obtained in the theorem. Next, the mean and variance of the SNR $\gamma^{\mathsf{st}}$ is computed as follows
	\begin{align}
		& \mathbb{E}\left\{\gamma^{\mathsf{st}} \right\} = \nu \mathbb{E}\left\{ {{C^2}} \right\} 
		= \nu{w_{c}^2}{k_{c}}\left( {{k_{c}} + 1} \right),
		 \\
		&\mathsf{Var} \left\{ \gamma^{\mathsf{st}} \right\} = 
		\nu^2 \left( {\mathbb{E}\left\{ {{C^4}} \right\} - {\left| {\mathbb{E}\left\{ {{C^2}} \right\}} \right|^2}} \right) \notag \\
  & = 2\nu^2{w_{c}^4}{k_{c}}\left( {{k_{c}} + 1} \right)\left( {2{k_{c}} + 3} \right) ,
	\end{align}
	where $\mathbb{E}\left\{ {{C^m}} \right\} = \frac{{{ w_{c}^m}\Gamma \left( {m + {k_{c}}} \right)}}{{\Gamma \left( {{k_{c}}} \right)}}$.
	Having obtained the mean and variance of $\gamma^{\mathsf{st}}$, we apply the second matching as follows
	\begin{equation} \label{eq:kw}
		k^{\mathsf{st}} = \frac{\left| \mathbb{E} \left\{ \gamma^{\mathsf{st}} \right\} \right|^2}{\mathsf{Var}\left\{ \gamma^{\mathsf{st}} \right\}}, w^{\mathsf{st}} = \frac{{\mathsf{Var}\left\{ \gamma^{\mathsf{st}} \right\}}}{{ \mathbb{E} \left\{ \gamma^{\mathsf{st}} \right\}}}.
	\end{equation}
By matching $\gamma^{\mathsf{st}}$ to a Gamma distribution with the shape and scale parameters in \eqref{eq:kw}, we obtain the coverage probability as shown in the theorem.

	\subsection{The Derivation of \eqref{eq:1stPcov} with the Short-Term Phase Shift Design} \label{Appendix:1stPcov}
	The first derivative of the coverage probability with respect to the coordinate $\varrho_{\mathrm{r}}$ of the RIS is provided as follows:
	\begin{align}
	 P_{\mathsf{cov}}^{\mathsf{st}} \left( \varrho_{\mathrm{r}} \right) = \frac{{ N \left( \varrho_{\mathrm{r}}  \right)}}{{D\left( \varrho_{\mathrm{r}}  \right)}} = \frac{{ \dt{N}\left( \varrho_{\mathrm{r}}  \right)D\left( \varrho_{\mathrm{r}}  \right) - N\left( \varrho_{\mathrm{r}}  \right) \dt{D}\left( \varrho_{\mathrm{r}}  \right)}}{{{{ D^2 \left( \varrho_{\mathrm{r}}  \right)}}}},
	\end{align}
where $\dt{N}\left( \varrho_{\mathrm{r}}  \right)$ is computed by using the formulation in \eqref{eq:NDDef} and the upper incomplete Gamma function as follows
\begin{equation} \label{eq:1stNx}
\begin{split}
		& \dt{N}\left( \varrho_{\mathrm{r}} \right) = \frac{{\partial {\Gamma \left( {{k^{\mathrm{st}}}\left( \varrho_{\mathrm{r}} \right),u\left( \varrho_{\mathrm{r}} \right)} \right)}}}{{\partial \varrho_{\mathrm{r}} }}  \mathop = \limits^{\left( a \right)}  - \dt{u}\left( \varrho_{\mathrm{r}} \right)u{\left(\varrho_{\mathrm{r}} \right)^{k^{\mathsf{st}} \left( \varrho_{\mathrm{r}} \right) - 1}} \times \\
		& \exp ( - u ( \varrho_{\mathrm{r}}) )  + \dt{k}^{\mathsf{st}} ( \varrho_{\mathrm{r}} )\int\nolimits_{t = u ( \varrho_{\mathrm{r}} )}^\infty  {t^{k^{\mathsf{st}} \left( \varrho_{\mathrm{r}} \right) - 1}}\log \left( t \right)\exp \left( { - t} \right)dt\\ 
  & = - \dt{u}\left( \varrho_{\mathrm{r}} \right) I_1 (\varrho_{\mathrm{r}}) + \dt{k}_n^{\mathsf{st}}\left( \varrho_{\mathrm{r}} \right) I_2 (\varrho_{\mathrm{r}}),
\end{split}
\end{equation}
where  $u(\varrho_{\mathrm{r}}) = z /w^{\mathsf{st}}(  \varrho_{\mathrm{r}})$, so its first-order derivative with respect to $\varrho_{\mathrm{r}}$ is  $\dt{u} ( \varrho_{\mathrm{r}} ) =  - z ( {w^{\mathsf{st}} ( \varrho_{\mathrm{r}})} )^{-2} \dt{w}^{\mathsf{st}}( \varrho_{\mathrm{r}} )$. Moreover, the functions $I_1( \varrho_{\mathrm{r}})$ and $I_2( \varrho_{\mathrm{r}})$ are formulated as
\begin{align}
		& I_1 ( \varrho_{\mathrm{r}}) = u{( \varrho_{\mathrm{r}} )^{k^{\mathsf{st}}( \varrho_{\mathrm{r}}) - 1}}\exp( { - u ( \varrho_{\mathrm{r}} )} )
		\nonumber \\
		& I_2 ( \varrho_{\mathrm{r}})= \Gamma \left( {k_n^{\mathsf{st}} \left( \varrho_{\mathrm{r}} \right),u\left( \varrho_{\mathrm{r}} \right)} \right)\log \left( {u\left( \varrho_{\mathrm{r}} \right)} \right)
		+  \Gamma \left( {k\left( \varrho_{\mathrm{r}} \right)} \right) \times \notag \\
		& \left( { - \log \left( {u\left( \varrho_{\mathrm{r}} \right)} \right) + \psi \left( {k\left( \varrho_{\mathrm{r}} \right)} \right)} \right) + u{\left( \varrho_{\mathrm{r}} \right)^{k_n^{\mathsf{st}} \left( \varrho_{\mathrm{r}} \right)}}{\left( {k^{\mathsf{st}} \left( \varrho_{\mathrm{r}} \right)} \right)^{ - 2}} \times \notag \\
		& {\;_2}{F_2}\left( {k^{\mathsf{st}}\left( \varrho_{\mathrm{r}} \right),k^{\mathsf{st}} \left( \varrho_{\mathrm{r}} \right),1 + k^{\mathsf{st}} \left( \varrho_{\mathrm{r}} \right),1 + k^{\mathsf{st}} \left( \varrho_{\mathrm{r}} \right), - u\left( \varrho_{\mathrm{r}} \right)} \right)
	\end{align}
In \eqref{eq:1stNx}, 
	$\left( a \right)$ is obtained by using Leibniz's integral rule for the upper incomplete Gamma function. For further processing, let us reformulate  $w^{\mathsf{st}} ( \varrho_{\mathrm{r}} )$ and $k^{\mathsf{st}} ( \varrho_{\mathrm{r}} )$ as
	\begin{align}
		w^{\mathsf{st}} ( \varrho_{\mathrm{r}} ) =& \frac{{2\nu{l_2} ( \varrho_{\mathrm{r}}) }}{{{l_1} ( \nu_{\mathrm{r}} ) }}, 
		k_n^{\mathsf{st}} ( \varrho_{\mathrm{r}} ) = \frac{{{ l_1^2 ( \varrho_{\mathrm{r}} ) }}}{{2  {l_2} ( \varrho_{\mathrm{r}} )}},
	\end{align}
where $l_1(\varrho_{\mathrm{r}})$ and  $l_2(\varrho_{\mathrm{r}})$ are given as
\begin{align}
& {l_1}( \varrho_{\mathrm{r}} ) = { w_{c}^2 ( \varrho_{\mathrm{r}} )}{k_{c}} ( \varrho_{\mathrm{r}} )( {1 + k_{c} ( \varrho_{\mathrm{r}})} ), \\
& {l_2} ( \varrho_{\mathrm{r}} ) = {w_{c}^4 (  \varrho_{\mathrm{r}} ) }{k_{c}} (  \varrho_{\mathrm{r}} ) ( {{k_{c}} (  \varrho_{\mathrm{r}}) + 1} ) ( {2{k_{c}} (  \varrho_{\mathrm{r}} ) + 3} ).
\end{align}
The first derivative of $w^{\mathsf{st}} ( \varrho_{\mathrm{r}} )$ and $k^{\mathsf{st}} ( \varrho_{\mathrm{r}})$ with respect to $\varrho_{\mathrm{r}}$ are  computed as 
	\begin{align}
		 \dt{w}^{\mathsf{st}} ( \varrho_{\mathrm{r}} ) =&
		{\left( \nu{l_1}(  \varrho_{\mathrm{r}}) \right)^{ - 2}}
		\left( {2{\nu^3}{\dt{l}_2}( \varrho_{\mathrm{r}}) }{l_1} (\varrho_{\mathrm{r}})  
		-  {2{\nu^3 }{l_2} ( \varrho_{\mathrm{r}} ) } 
	 \dt{l}_1 ( \varrho_{\mathrm{r}} )  
		\right), \label{eq:dwst} \\
		{\dt{k}^{\mathsf{st}}}  (  \varrho_{\mathrm{r}} ) =& 
		{\left( 2{\nu^2}{l_2} ( \varrho_{\mathrm{r}} )  \right)^{ - 2}}
		\left(
		4 {\nu^4{l_1} (\varrho_{\mathrm{r}})  } 
	 {\dt{l}_1} ( \varrho_{\mathrm{r}} ) 
		 {l_2} ( \varrho_{\mathrm{r}} ) \right. \notag \\ 
   & \left.  
		-  2{ \nu^4}{\dt{l}_2}( \varrho_{\mathrm{r}} )  
		l_1^2 ( \varrho_{\mathrm{r}} )
		\right), \label{eq:dkst}
	\end{align}
	where  ${\dt{l}_1} ( \varrho_{\mathrm{r}} )$ and ${\dt{l}_2} ( \varrho_{\mathrm{r}} )$ are the first derivative of $	{l_1} ( \varrho_{\mathrm{r}} )$ and ${l_2} ( \varrho_{\mathrm{r}} )$ with respect to $ \varrho_{\mathrm{r}}$, which are computed as follows
	\begin{align}
		{\dt{l}_1} ( \varrho_{\mathrm{r}} ) =& 
		2{w_{c} } ( \varrho_{\mathrm{r}} ){\dt{w}_{c}}( \varrho_{\mathrm{r}}){k_{c}} ( \varrho_{\mathrm{r}} )\left( {1 + {k_{c}}\left(  \varrho_{\mathrm{r}} \right)} \right) + \notag  \\
  & { {{w_{c}^2}\left( \varrho_{\mathrm{r}} \right)}}\left( {1 + 2{k_{c}}\left( \varrho_{\mathrm{r}} \right)} \right){\dt{k}_{c}}\left( \varrho_{\mathrm{r}} \right), \\
		{\dt{l}_2}\left(  \varrho_{\mathrm{r}} \right) =& 
		4{w_{c}^3\left(  \varrho_{\mathrm{r}} \right)}{\dt{w}_{c}}\left(  \varrho_{\mathrm{r}} \right){k_{c}}\left(  \varrho_{\mathrm{r}} \right)\left( {{k_{c}}\left( x \right) + 1} \right)\left( {2{k_{c}}\left(  \varrho_{\mathrm{r}} \right) + 3} \right)
		\nonumber \\
		& + {{w_{c}^4 \left(  \varrho_{\mathrm{r}} \right)}}\left( {6{{\left( {{k_{c}}\left(  \varrho_{\mathrm{r}} \right)} \right)}^2} + 10{k_{c}}\left(  \varrho_{\mathrm{r}} \right) + 3} \right){\dt{k}_{c}}\left(  \varrho_{\mathrm{r}} \right),
	\end{align}
where $\dt{w}_{c} ( \varrho_{\mathrm{r}} )$ and $\dt{k}_{c} ( \varrho_{\mathrm{r}}  )$ are the first derivative of the scale and shape parameters $w_{c} ( \varrho_{\mathrm{r}} )$ and $k_{c} ( \varrho_{\mathrm{r}}  )$ with respect to $ \varrho_{\mathrm{r}}$, respectively, and defined as follows
	\begin{align} \label{Eq:Apn_Deri_Opt:04}
		& \dt{w}_{c}(  \varrho_{\mathrm{r}} ) = \frac{ \dt{t} ( \varrho_{\mathrm{r}}) I_3(\varrho_{\mathrm{r}}) }{{{{\left( {{c_{1}} + {\tilde{c}_{2}}\sqrt {t\left(  \varrho_{\mathrm{r}} \right)} } \right)}^2}}},  \dt{k}_c\left(  \varrho_{\mathrm{r}} \right) = \frac{\dt{t}\left(  \varrho_{\mathrm{r}} \right){{I_4(\varrho_{\mathrm{r}})}}}{{{{\left( {c_{3} + {\tilde{c}_{4}}t (  \varrho_{\mathrm{r}})} \right)}^2}}},
	\end{align}
	where $ {c_{2}} =  {\tilde{c}_{2}}\sqrt {t (  \varrho_{\mathrm{r}} )}$ and ${c_{4}} = {\tilde{c}_{4}}t (  \varrho_{\mathrm{r}} ) $. Besides, the supplementary parameters $t(\varrho_{\mathrm{r}})$, $\dt{t}(\varrho_{\mathrm{r}})$, $I_3$, and $I_4$ are given as follows
	\begin{align}
	  I_3(\varrho_{\mathrm{r}}) =&  c_{4} \left( {{c_{1}} + {\tilde{c}_{2}}\sqrt {t\left( \varrho_{\mathrm{r}} \right)} } \right) - {0.5}{ {t^{-0.5}\left( \varrho_{\mathrm{r}} \right)}}\left( {{c_{3}} + {\tilde{c}_{4}}t\left( \varrho_{\mathrm{r}} \right)} \right){c_{2}}, \\
	 I_4(\varrho_{\mathrm{r}})=&  \left( {c_{1} + \tilde{c}_{2}\sqrt {t (\varrho_{\mathrm{r}})} } \right) \left( \tilde{c}_{2}{ {t^{-1/2}\left( \varrho_{\mathrm{r}}\right)}}\left( c_{3} + {\tilde{c}_{4}}t( \varrho_{\mathrm{r}}) \right) \right. \notag \\ 
   &\quad \left. - {\tilde{c}_{4}}\left( c_{1} + {\tilde{c}_{2}}\sqrt {t( \varrho_{\mathrm{r}})}  \right) \right),\\
	  t( \varrho_{\mathrm{r}} ) =& {\beta_{\mathrm{s_1r}}}\left( \varrho_{\mathrm{r}} \right){\beta_{\mathrm{rd}}}\left( \varrho_{\mathrm{r}} \right), \\
	 \dt{t}\left( \varrho_{\mathrm{r}} \right) = & {\dt{\beta}_{\mathrm{s_1r}}}\left( \varrho_{\mathrm{r}} \right){\beta_{\mathrm{rd}}}\left( \varrho_{\mathrm{r}} \right) + {\beta_{\mathrm{s_1r}}}\left( \varrho_{\mathrm{r}} \right){\dt{\beta}_{\mathrm{rd}}}\left( \varrho_{\mathrm{r}} \right), \\
		 {\beta_{\mathrm{sr}}}\left( \varrho_{\mathrm{r}} \right) =& {K_{0}}\left( {{{\left( {x_{\mathrm{r}} - {x_{\mathrm{s}}}} \right)}^2} + {{\left( {y_{\mathrm{r}} - {y_{\mathrm{s}}}} \right)}^2}}  + {{\left( {{z_{\mathrm{r}}}  - {z_{\mathrm{s}}}} \right)}^2} \right)^{-\eta_{\mathrm{sr}} /2}, \\
		 {\dt{\beta}_{\mathrm{sr}}}\left( \varrho_{\mathrm{r}} \right) =& -\eta_{\mathrm{sr}} {K_{0}}\left( {\varrho_{\mathrm{r}} - {v_{\mathrm{s}}}} \right) 
		 \left( {{{\left( {{x_{{\mathrm{r}}}} - {x_{\mathrm{s}}}} \right)}^2} + {{\left( {y_{\mathrm{r}} - {y_{\mathrm{s}}}} \right)}^2} }  + \right. \notag \\
   & \quad \left. {{\left( {{z_{\mathrm{r}}} - {z_{\mathrm{s}}}} \right)}^2} \right)^{-\eta_{\mathrm{sr}} /2 - 1},
\end{align}
 \begin{align}
		{\beta_{\mathrm{rd}}}\left( \varrho_{\mathrm{r}} \right) =&  {K_{0}}{\left( {{{\left( {x_{\mathrm{r}} - {x_{\mathrm{d}}}} \right)}^2} + {{\left( {y_{\mathrm{r}} - {y_{\mathrm{d}}}} \right)}^2}} + {{\left( {{z_{\mathrm{r}}} - {z_{\mathrm{d}}}} \right)}^2} \right)^{-\eta_{\mathrm{rd}} /2}}, \\
		{\dt{\beta}_{\mathrm{rd}}}\left( \varrho_{\mathrm{r}} \right) =&  -\eta_{\mathrm{rd}} {K_{0}} \left( {\varrho_{\mathrm{r}} - {\alpha_{\mathrm{d}}}} \right)\left( {{{\left( {{x_{{\mathrm{r}}}} - {x_{\mathrm{d}}}} \right)}^2} + {{\left( {y_{\mathrm{r}} - {y_{\mathrm{d}}}} \right)}^2} } + \right. \notag \\
   & \quad \left. {{\left( {{z_{\mathrm{r}}} - {z_{\mathrm{d}}}} \right)}^2} \right)^{-\eta_{\mathrm{rd}} /2 - 1},
	\end{align}
	where $\alpha_\mathrm{d} \in \{ x_\mathrm{d},  y_\mathrm{d},  z_\mathrm{d} \}$. In a similar manner, the derivative of $ D\left( \varrho_{\mathrm{r}} \right) = \Gamma \left( {k^{\mathsf{st}}\left( \varrho_{\mathrm{r}} \right)} \right) $ in the numerator of \eqref{eq:1stPcov} with respect to $\varrho_{\mathrm{r}}$ is computed as follows:
	\begin{equation}
		\dt{D}\left( \varrho_{\mathrm{r}} \right) = \dt{k}^{\mathsf{st}} \left(\varrho_{\mathrm{r}} \right)\Gamma \left( {k^{\mathsf{st}}\left(\varrho_{\mathrm{r}} \right)} \right)\psi \left( {k^{\mathsf{st}}\left( \varrho_{\mathrm{r}} \right)} \right),
	\end{equation}
where $\psi ( {k\left( x \right)} )$ is the polygamma function  of the first order. 
	%
	
	
	\subsection{The Derivation of \eqref{eq:1stPcov} with the Long-Term Phase Shift Design} \label{Appendix:1stPcovv1}
	The proof follows similar steps as for the short-term phase shift design. Specifically, the first-order derivative of the coverage probability with respect to $\varrho_{\mathrm{r}}$ is given in \eqref{eq:1stPcov} with 
	$\dt{N}\left( \varrho_{\mathrm{r}} \right)$ and $\dt{D}\left( \varrho_{\mathrm{r}} \right)$ defined in \eqref{eq:NDDef} that have the same structure as the short-term phase shift design, except for the scale and shape parameters. Let us therefore compute the derivative of $w_{n}^{\mathrm{lt}} ( \varrho_{\mathrm{r}} )$ and $k^{\mathrm{lt}} ( \varrho_{\mathrm{r}} )$ as a function of the RIS coordinates. To this end, we first reformulate the shape and scale parameters in \eqref{eq:klt} and \eqref{eq:wlt} to the corresponding equivalent forms
	\begin{align}
		k_{n}^{\mathrm{lt}} ( \varrho_{\mathrm{r}} ) = \mathsf{Nu}^2( \varrho_{\mathrm{r}})/\mathsf{De}( \varrho_{\mathrm{r}}), w^{\mathrm{lt}} (\varrho_{\mathrm{r}} ) = \mathsf{De}(\varrho_{\mathrm{r}}) /\mathsf{Nu}( \varrho_{\mathrm{r}}),
	\end{align}
where the following definitions hold for $\mathsf{Nu}( \varrho_{\mathrm{r}})$ and $\mathsf{De}( \varrho_{\mathrm{r}})$ as follows
\begin{align}
& \mathsf{Nu}( \varrho_{\mathrm{r}}) =	\nu \beta_{\mathrm{sd}} + 
	\nu \tilde{o}_{1} \beta_{\mathrm{sr}} ( \varrho_{\mathrm{r}}) \beta_{\mathrm{rd}}  (\varrho_{\mathrm{r}}) , \\
& \mathsf{De}( \varrho_{\mathrm{r}}) = \nu^2 \beta_{\mathrm{sd}}^2 +
\nu^2 \tilde{o}_{2} \beta_{\mathrm{sr}}^2 ( \varrho_{\mathrm{r}}) \beta_{\mathrm{rd}}^2  (\varrho_{\mathrm{r}}) + 2 \nu^2 \beta_{\mathrm{sd}} \tilde{o}_{1} \beta_{\mathrm{sr}} ( \varrho_{\mathrm{r}}) \beta_{\mathrm{rd}}  (\varrho_{\mathrm{r}}) ,
\end{align}
with $o_{1} = \tilde{o}_{2}  \beta_{\mathrm{s_1r}} ( \varrho_{\mathrm{r}}) \beta_{\mathrm{rd}}  (\varrho_{\mathrm{r}})$ and $o_{2} = \tilde{o}_{2}  \beta_{\mathrm{sr}}^2 ( \varrho_{\mathrm{r}}) \beta_{\mathrm{rd}}^2  (\varrho_{\mathrm{r}})$. The first-order derivative of the shape and scale parameters with respect to $\varrho_{\mathrm{r}}$ are then computed as follows
\begin{align}
	&\dt{k}^{\mathrm{lt}} ( \varrho_{\mathrm{r}} ) 	= \mathsf{De} ( \varrho_{\mathrm{r}} )^{-2} \left( 2  \mathsf{Nu}  ( \varrho_{\mathrm{r}} ) 
		\dt{\mathsf{Nu}} ( \varrho_{\mathrm{r}}  ) \mathsf{De} ( \varrho_{\mathrm{r}}  )
		-  \mathsf{Nu}^2 ( \varrho_{\mathrm{r}} )  \dt{\mathsf{De}} ( \varrho_{\mathrm{r}} )
		\right), \\
		& \dt{w}_{n}^{\mathrm{lt}} ( \varrho_{\mathrm{r}} ) 
		= \mathsf{Nu}^{-2} ( \varrho_{\mathrm{r}} ) \left(  
		\dt{\mathsf{De}} ( \varrho_{\mathrm{r}}  ) \mathsf{Nu} ( \varrho_{\mathrm{r}} )
		- \mathsf{De} ( \varrho_{\mathrm{r}} ) \dt{\mathsf{Nu}} ( \varrho_{\mathrm{r}}  )
		\right), \\
		& \dt{\mathsf{Nu}} ( \varrho_{\mathrm{r}} ) = \nu \tilde{o}_{1} \dt{t} (  \varrho_{\mathrm{r}} ),\\
		& \dt{\mathsf{De}} (\varrho_{\mathrm{r}}) =2 
		\nu^2 \tilde{o}_{2}  t (  \varrho_{\mathrm{r}} ) \dt{t} ( \varrho_{\mathrm{r}} ) + 2 \nu^2 \beta_{\mathrm{sd}} \tilde{o}_{1} \dt{t} ( \varrho_{\mathrm{r}} ),
	\end{align}
	where $\dt{t} ( v_{\mathrm{r},n}  )$ is provided in \eqref{Eq:Apn_Deri_Opt:04} and we, therefore, complete the proof.
	\bibliographystyle{IEEEtran}
	\bibliography{IEEEabrv,refs}
\end{document}